\documentclass[10pt]{IEEEtran}
\usepackage{amssymb}
\usepackage{delarray}
\usepackage{graphicx}
\newtheorem{Corollary}{Corollary}

\newtheorem{proposition}{Proposition}

\newtheorem{remark}{Remark}
\newtheorem{proof}{Proof}
\usepackage{url}
\usepackage{multicol}
\usepackage{algorithmic}
\usepackage{algorithm}
\usepackage{color}
\usepackage{epstopdf}

\begin{document}
\newcommand{\beq}{\begin{equation}}
\newcommand{\deq}{\end{equation}}

\newcommand{\baq}{\begin{eqnarray}}
\newcommand{\daq}{\end{eqnarray}}

\newcommand{\baqm}{\begin{eqnarray*}}
\newcommand{\daqm}{\end{eqnarray*}}

\newcommand{\up}[1]{\overline{#1}}
\newcommand{\down}[1]{\underline{#1}}

\newcommand{\ti}[1]{\tilde{#1}}
\newcommand {\R} {\rm I\!R}
\newcommand{\qed}{\rightline{$\Box$}}
\newcommand{\dfn}{\stackrel{\Delta}{=}}
\newcommand{\N}{I\!\!N}
\newcommand{\Pro}{I\!\!P}
\newcommand{\E}{I\!\!E}
\newcommand{\one}{{\bf 1}}
\newcommand{\range}[2]{#1,\dots,#2}

\newcommand{\bn}{\binom}

\newcommand{\eps}{\varepsilon}
\newenvironment{pf}{{\bf Proof. }}{\hfill $\square$\medskip}
\newtheorem{thm}{Theorem}[section]
\newtheorem{ass}{Assumption}[section]
\newtheorem{defi}{Definition}[section]
\newtheorem{prop}{Proposition}[section]
\newtheorem{cor}{Corollary}[section]
\newtheorem{lemma}{Lemma}[section]
\newtheorem{rem}{Remark}[section]
\newtheorem{algo}{Algorithm}[section]
\def\R{\mathop{\rm I\kern -0.20em R}\nolimits}
\newtheorem{condition}{Condition}

\newcommand{\tprod}{\otimes}
\newcommand{\bigtprod}{\bigotimes}
\newcommand{\tprodcon}{\odot}
\newcommand{\bigtprodcon}{\bigodot}

\title{Optimal Energy-Delay Tradeoff for Opportunistic Spectrum Access in Cognitive Radio Networks}

\author{Oussama Habachi*\thanks{*Corresponding author: oussama.habachi@unilim.fr}, Yezekael Hayel and Rachid El-Azouzi\\CERI/LIA, University of Avignon, Agroparc BP 1228, Avignon, France}

\maketitle

\begin{abstract}
Cognitive radio (CR) has been considered as a promising technology to enhance spectrum efficiency
via opportunistic transmission at link level. Basic CR features allow SUs to transmit only when the licensed primary
channel is not occupied by PUs. However, waiting for idle time slot may include large packet
delay and high energy consumption. We further consider that the SU may decide, at any moment, to use another dedicated way  of communication (3G) in order to transmit its packets. Thus, we consider an Opportunistic Spectrum Access (OSA) mechanism that takes into account packet delay and energy consumption.  We formulate the OSA problem as a Partially Observable Markov Decision Process (POMDP) by explicitly considering the energy consumption as well as packets' delay, which are often ignored in existing OSA solutions. Specifically, we consider a POMDP with an average reward criterion.   We derive structural properties of the value function and we show the existence of optimal strategies in the class of the threshold strategies.
For implementation purposes, we propose online learning mechanisms that estimate the PU activity  and determine the appropriate threshold strategy on the fly. In particular, numerical illustrations validate our theoretical findings.

\end{abstract}
\section{Introduction}
The access to  spectrum frequencies is defined by licenses assigned to PUs. The latter must be conform to the specifications described in the license (e.g. location of the base station, frequency and the maximum transmission power). Nonetheless, a recent study made by the Federal Communications Commission (FCC) has proved that some frequency bands are not sufficiently used by licensed users at a particular time and in a specific location \cite{EkHoss}.

Cognitive radio, which is a new paradigm for designing wireless communication systems, has appeared in order to enhance the utilization of the radio frequency spectrum. It has been considered as the key technology that enable SUs to access the licensed spectrum. A cognitive user, as defined in \cite{mitola}, is a mobile who has the faculty to adapt its transmission parameters (e.g. frequency and modulation) to the wireless environment, and support different communication standards (e.g. GSM, CDMA, WiMAX and  WiFi). Moreover, when there is no opportunity to transmit over licensed primary channels,  SUs may have the possibility to transmit on dedicated channels, generally, with a higher cost and/or a lower throughput than transmitting over licensed primary channels. The possibility of having dedicated channels reserved for secondary mobiles has been proposed in \cite{Ak06} and \cite{Jaganathan}.


{
In this paper, we develop a threshold-based OSA for SUs, taking into account the energy and the delay, that can be applied in different cognitive radio context. For example, as we can see in Figure \ref{adhoc}, a SU-Tx, i.e. transmitter, is equipped with two transceivers (a Software Defined Radio (SDR) transceiver to sense and access the licensed spectrum and a control transceiver to notify the SU-Rx, i.e. receiver,  by the channel that will be used for the transmission). SU-Tx communicates with other SUs through an ad-hoc connection using a spectrum hole of a licensed frequency. This  scenario was studied in \cite{Su08}. The model, that we consider in our paper, is also suited for the scenario depicted in Figure \ref{CRnetwork} where the SU is a cognitive base station which is able to sense the activity of a primary base station, and then takes profit of spectrum holes for transmitting on the downlink. Indeed, the SU is a base station which uses the licensed spectrum to transmit to its users when the spectrum is not used by the primary base station.}
\begin{figure}
\begin{center}pdf
  \includegraphics[width=0.35\textwidth]{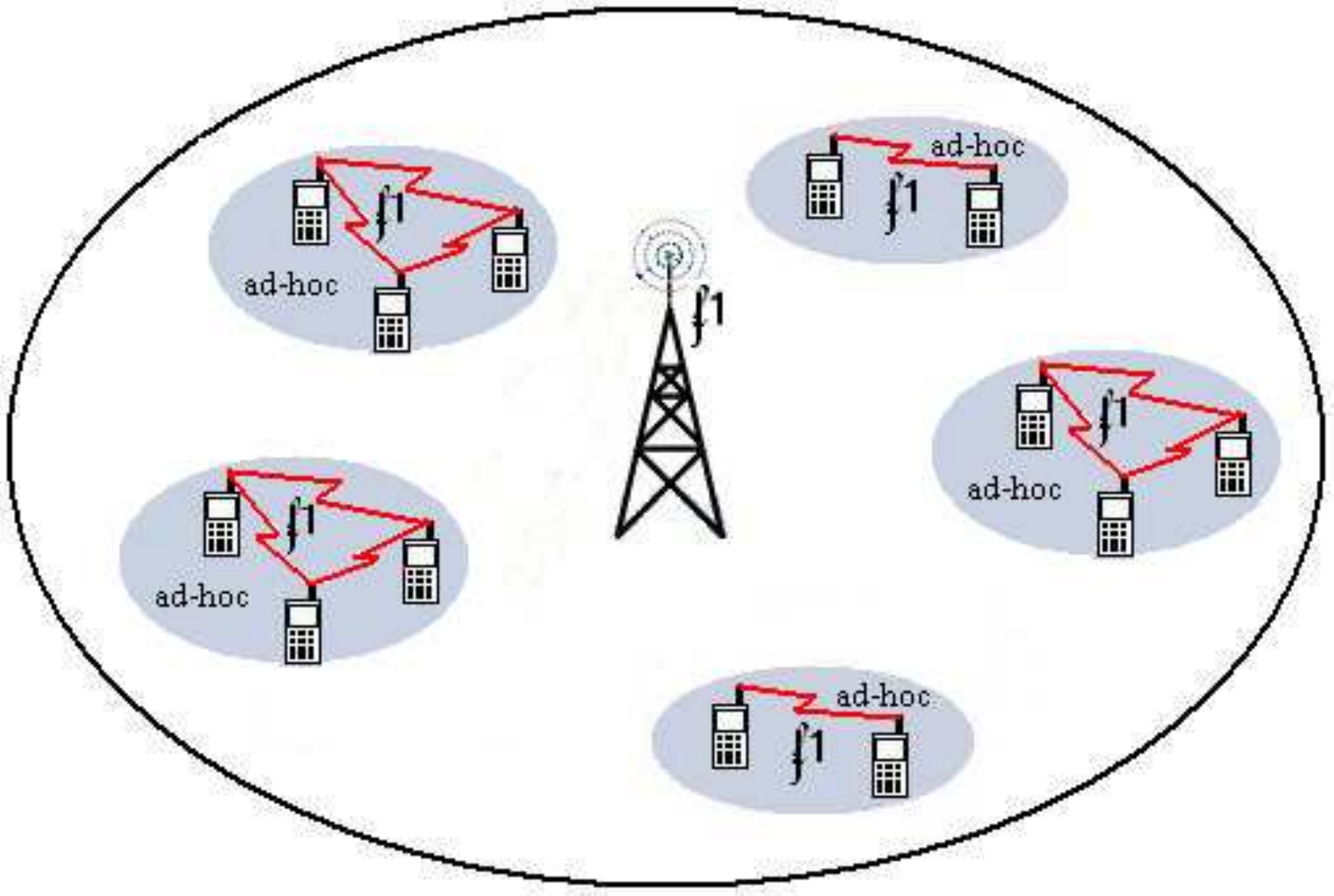}
  \caption{Using cognitive radio in ad-hoc communication. If the licensed frequency $f1$ is not used by PUs, SUs can communicate in ad-hoc mode using $f1$.}\label{adhoc}
  \end{center}
\end{figure}
\begin{figure}
\begin{center}
  \includegraphics[width=0.35\textwidth]{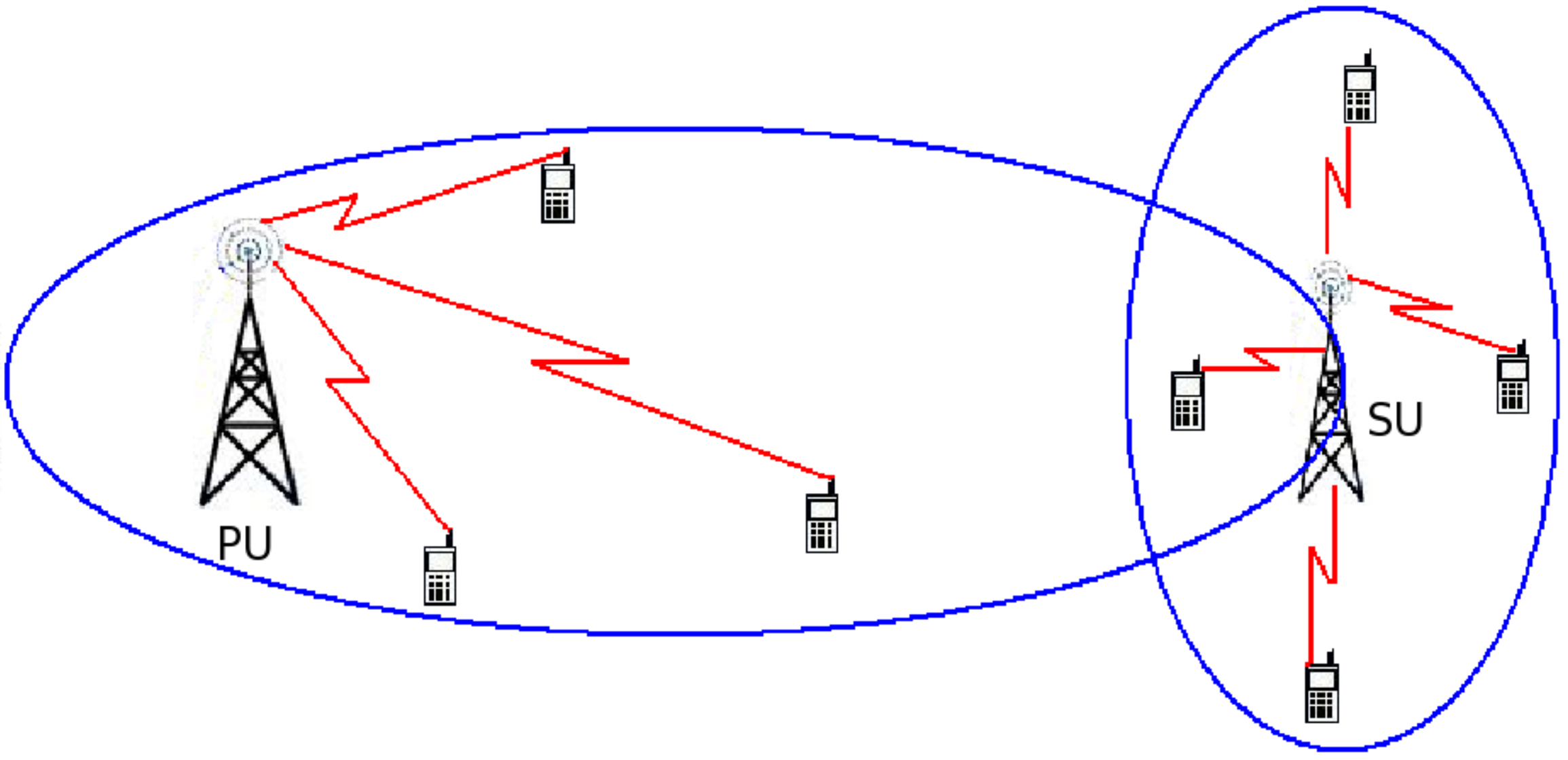}
  \caption{SU is a cognitive base station which is able to sense the activity of a PU base station, and then takes profit of spectrum holes.}\label{CRnetwork}
  \end{center}
\end{figure}
Our main contribution is to consider in this cognitive radio setting, an optimal opportunistic spectrum access (OSA) mechanism that takes into account energy consumption and packets' delay. Many works have focused on the study of optimal sensing and access policies in cognitive radio networks (see \cite{Qing}, \cite{Hua08} and \cite{zheng05}).  All these works have focused on either spectrum sensing or dynamic spectrum sharing.  In \cite{Shi12}, the authors study power control, scheduling and routing problems for maximizing the data rates of SU in multi hop Cognitive Radio Networks. In \cite{Min11}, mobility aspects of the users, both SU and PU, is considered for determining an optimal OSA. In \cite{energy2}, the authors focused on an OSA problem with an energy constraint. The authors have formulated their problem as a POMDP and derived some properties of the optimal sensing control policies. Their control parameter is the duration of sensing used by a SU at each time slot for determining the PU activity. They provided heuristic control policies using a gird-based approximation, myopic policies and static policies which have low complexity but give suboptimal control policies. Authors of \cite{Yu08} incorporate the energy constraint in the design of the optimal policy of sensing and access in cognitive radio network. They formulate the problem also as a POMDP  with a finite horizon and established a threshold structure of the optimal policy for the single channel model.   \cite{Sultan} characterised  the optimal sensing and access for a SU  with an energy queue. It is noteworthy that the impact of the energy consumption or the capacity of cognitive radio to support additional Quality-of-Service (QoS), such as the expected delay, has been somehow ignored in the literature.

{The slow advance of battery technology for mobile devices has motivated both academic and industry to focus on energy efficient transmission in order to create a more satisfactory user experience (see \cite{Saaverdra14}, \cite{Xiong14}, \cite{wu09}, and \cite{Pei11}). Authors of \cite{Pei11} considered that a SU senses sequentially some licensed primary channels before deciding to start transmission. They studied the sensing order and strategy, and the power allocation for a single pair of SU transmitter (SU-Tx and SU-Rx). In \cite{wu09}, the authors  considered an energy-efficient transmission for CR with a delay constraint. Similarly to our work, they  considered an objective function that incorporate a cost for both the consumed energy and the delay. They assumed that the SU senses the licensed channel at the beginning of the slot in order to estimate the activity of the PU as well as the channel power gain.  Hence, the   authors formulated the problem as a discrete-time Markov decision process (MDP) in order to minimize the delay and the energy costs when transmitting the target payload.   The energy of sensing is not considered in the energy consumption that the SU aims to minimize with the delay cost.    

Multiuser opportunistic spectrum access has been investigated extensively in the past years.  In \cite{Saaverdra14}, the authors tried to maximize the energy efficiency in a wireless network with multiple contending nodes using distributed opportunistic scheduling. They tuned the performance of the system using the access probability and the threshold rates.  In \cite{Xiong14}, the authors studied energy efficient opportunistic spectrum access strategies for an orthogonal frequency division multiplexing OFDM-based CR networks with multiple SUs, where each subchannel is exclusively assigned to at most one SU to avoid interference among different SUs.} We have addressed the multiuser problem in our previous paper \cite{oussama}  using game theory and Partially Observable Stochastic Game (POSG), a multiuser version of the POMDP. We illustrated the existence of a tradeoff between large packet delay, partially due to collisions between SUs, and high energy consumption.

Without considering the packets' delay, the SU achieves the best tradeoff between trying to access the licensed primary channel and sleeping to conserve energy.
In fact, it is very important for today multimedia applications on wireless networks, to provide reliable communication while sustaining a certain level of QoS. Moreover, taking into account the transmission delay as well as the energy consumption significantly complicates the optimization problem. The design of such tradeoff lies among several conflicting objectives: gaining immediate access, gaining spectrum occupancy information, conserving energy and minimizing packets' delay. Then, the goal of our paper is to study such energy-QoS tradeoff for determining an optimal OSA mechanism for SUs in a cognitive radio network. The major contributions of our work are:
\begin{itemize}
  \item Instead of improving  existent OSA mechanism, we consider an original more complicated problem that take into account the energy consumption as well as  packet delays. 
    \item The problem is formulated as an infinite horizon POMDP with average criterion. The average criterion is better than the discount or the total criterion as the SU takes often decisions. 
    \item In order to gain insights into the energy-delay constrained OSA problem, we derive structural properties of the value function. We  show that the value function is increasing with the belief and decreasing with  packet delays. These structural results not only give us the fundamental  thresholds design, but also reduce the computational complexity when seeking for the optimal policies.
    \item We show that the SU  maximizes its average reward by adopting a simple threshold policy, and we derive closed-form expressions for these thresholds. The instantaneous reward is defined as a function of the gain (number of bits transmitted) and costs (transmission costs, sensing costs and delay).
    \item Since the SU may use a dedicated channel for its packets, the optimal threshold policy guarantees a bounded delay.
\end{itemize}

\begin{table}[h]
\label{tb:symbols}
\caption{Table of symbols}
\begin{center}
\begin{tabular}{|p{1cm}|p{5cm}|}\hline
$n$ &  wireless channel \\\hline
 $n^*$ &  wireless channel chosen for sensing \\\hline
$s_n(t)$ & state of channel $n$ at time $t$\\\hline
$\alpha_n, \beta_n$ & transition probabilities of primary user on channel n\\\hline
$\lambda_n(t)$ & belief probability of channel $n$ at time $t$  \\\hline
 $l(t)$ & delay of packet at time $t$ \\\hline
 $a(t)$ & action of SU at time $t$ \\\hline
$\theta(t)$ & observation of the SU at time $t$ \\\hline
 $\mu_t$ & strategy  of the SU at time $t$  \\\hline
  $\mu^*$ & the optimal policy  \\\hline
 $\phi$ & the reward \\\hline
 $c_s$ & sensing cost \\\hline
  $P_p$ & transmission cost over a licensed channel \\\hline
    $P_{3G}$ & transmission cost over dedicated access \\\hline
    $f(l)$ & delay penalty \\\hline
    $\pi_n(0)$ & the stationary probability that the licensed channel $n$ is in idle state \\\hline
\end{tabular}
\end{center}
\end{table}

The organization of the paper is as follows. In the next section, we describe the primary and the SU models. Section \ref{section3} presents our partially observable Markov decision process framework. In Section \ref{section4}, we study the existence of an optimal threshold policy for our opportunistic spectrum access with an energy-QoS tradeoff.  In Section \ref{seclearning}, we propose an online learning algorithm which can be used in practice by agents to solve the POMDP. Before concluding the paper and giving
some perspectives, we present, in Section \ref{section5}, some numeric illustrations.

\section{Cognitive radio network model}\label{section2}
We consider a wireless system with $N$ independent channels licensed to PUs. The state of each channel $n \in \{1,\ldots,N\}$ is modeled by a time-homogeneous discrete Markov process $s_n(t)$. The state space is $\{0,1\}$ where $s_n(t)=0$ means that the channel $n$ is free for SU access, and $s_n(t)=1$ means that the channel $n$ is occupied by PUs. The transition probabilities of the channel $n$ is given by the following matrix:
$$
P_n= \left(
\begin{array}{cc}
\alpha_n & 1-\alpha_n\\
\beta_n & 1-\beta_n
\end{array}
\right)
$$
The transition rates evolve as illustrated in Figure \ref{transi}.
The global system state, composed of the $N$ channels, is denoted by the vector $s(t)=[s_1(t),...,s_N(t)]$ and the global state space is $\mathcal{s}=\{0,1\}^N$. The transition probabilities can be determined by the statistics of the primary network traffic and are assumed to be known by SUs. We present in Section \ref{seclearning} how the SU can estimate these transition probabilities on the fly.
\begin{figure}[!h]
\centering
  \includegraphics[width=0.4\textwidth]{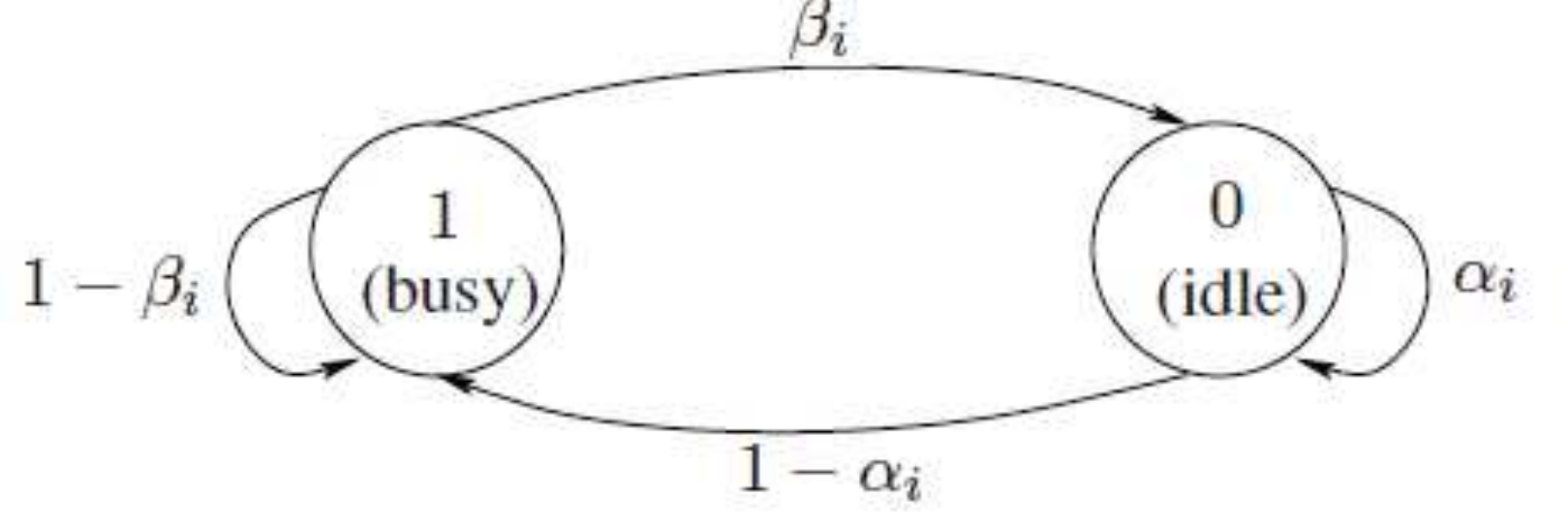}
  \caption{The channel transition probabilities for channel $i$.}\label{transi}
\end{figure}
We consider a SU having the possibility to access to anyone of the $N$ licensed primary channels. The objective of the SU is to detect the channels that are free during a given time slot. However, waiting for idle time slot may include large packet delay and high energy consumption due to the sensing. To overcome this, we consider an OSA that takes into account packet delay, throughput and energy consumption. Since today's wireless networks are highly heterogeneous with mobile devices consisting of multiple wireless network interfaces, we assume that at any time, the SU has access to the network through another technology like 3G. 
The SU will prefer to transmit its packet on a licensed primary channel because it is cheaper than a dedicated communication while the dedicated channel guarantees perfect access.

The goal of each SU is to minimize the expected delay of its packets, accounting for energy, throughput and monetary costs. One of our important contributions is to consider the average transmission delay of a packet in the optimal decision. Indeed, sensing a licensed primary channel has a cost for the SU. We look for an optimal sensing policy which depends on the history of observations and actions.

\section{Partially Observable Markov Decision Process Framework}\label{section3}

Due to partial spectrum sensing, the global system state $s(t)$ cannot be directly observed by a SU. To overcome this difficulty, the SU infers the global system state based on observations that can be summarized in a belief vector:
\begin{eqnarray*}
{\vec{\lambda}(t)}=[\lambda_1(t),..,\lambda_N(t)],
\end{eqnarray*}
where $\lambda_i(t)$ is the conditional probability that the channel $i$ is available at time slot $t$.  


We describe now the POMDP framework considered here.

\subsubsection{State}
The state of the system at time slot $t$ is given by $({\vec{\lambda}(t)},l(t))$ where $l(t)$ is the delay of the packet held by SU at time $t$. The delay of a new packet equals one, and increases by one every time slot, except when the SU transmits the packet.{ We consider a system without buffering , then the SU cannot handle a new packet until he transmits the packet in the system.  In this paper we consider the saturated case in which the SU  has always packets to transmit.}

\subsubsection{Action}
{A SU  chooses an action $a(t)\in\{0,1,2\}$ at each time slot from the following actions:
 \begin{itemize}
 \item $0$: Stay inactive during the time slot,
 \item $1$: Sense a licensed primary channel. If the channel is available transmit, otherwise wait for next time slot,
\item $2$:  Sense a licensed primary channel. If the channel is available transmit, otherwise use the dedicated channel.
\end{itemize}}

%

\subsubsection{Observation and belief}
When the SU decides to sense (i.e. to take action $a(t)\in\{1,2\}$), one channel $n^*$ is determined and the SU observes the channel occupancy state $s_{n^*}(t)\in\{0,1\}$.  Let  $\theta(t)$ be the observation outcome at time $t$, where $\theta(t)=0$ if the sensed channel is idle and  $\theta(t)=1$ otherwise. The user updates the belief vector $\vec{\lambda}(t)$ after the observation outcome. For each channel $n$, the conditional probability $\lambda_n(t+1):=\mbox{Pr}(s_n(t+1)=0|a(t),\theta(t))$ is  defined as follows:
{\small{
\begin{equation}\label{upomega}
\lambda_n(t+1)
=\left\{
\begin{array}{ll}
\beta_n+(\alpha_n-\beta_n)\lambda_n(t) & \mbox{if }a(t) = 0 \mbox{ or }  n\neq n^*, \\
 \alpha_n& \mbox{if} \quad a(t)\neq 0,\mbox{ }\theta(t)=0 \\&\mbox{ and }  n=n^*,\\
 \beta_n & \mbox{if} \quad a(t)\neq 0,\mbox{ } \theta(t)=1\\&\mbox{ and }  n=n^*.\\
\end{array}
\right.
\end{equation}}}

\subsubsection{Channel choice policy}
At each time slot $t$, based on its belief vector $\vec{\lambda}(t)$, the SU chooses a channel $n^*\in N$ to be sensed. There exists several channel choice policies in the literature such as deterministic, randomized and periodic (see \cite{EkHoss}). In this paper, we consider that the SU senses the channel which has the highest probability to be idle, i.e. $n^*:=\arg\max_{n}(\lambda_n(t))$.

\subsubsection{Policies}
The strategy of the SU is defined by the probability of choosing a given action depending on the system state. We define a sensing and access policy $\mu$ as a vector $[\mu_1,\mu_2,\ldots]$ where $\mu_t$ is a mapping from a state $(\vec{\lambda}(t),l(t))$ to an action $a(t)$. The set of policies is denoted by $\Gamma$. A stationary policy is a mapping that specifies for each state, independently of the time slot $t$, an action to be chosen. In the next section, we show that our POMDP problem has an optimal stationary policy which allows us to restrict our problem to the set of  stationary policies.

\subsubsection{Reward and costs}
\begin{itemize}
\item Reward:  Let $\Phi$ be the reward representing the number of delivered bits when the SU transmits its packet.
\item Sensing costs :  Let $c_s$ be the energy cost function for sensing a licensed channel. 
\item Transmission cost: The PU and the service provider for the dedicated access, charge a price for each packet transmitted. Those prices are respectively $P_p$ for a transmission over a primary channel and $P_{3G}$ for a transmission over the dedicated channel. 
\item Delay penalty: In order to model the impact of the delay, we introduce an additional cost when a packet is not transmitted. This cost depends on the current delay $l$ of the packet and it is defined by the function $f(l)$. This function is assumed to be increasing with $l$ in order to increase the incentive of transmitting the packet when it becomes delayed.
\end{itemize}

We have expressed all the rewards and cost in the same unit in order to achieve a tradeoff between energy and delay.

At time slot $t$, the instantaneous reward $r_t((\vec{\lambda}(t), l(t)),a(t))$ of a SU depends on the system state $(\vec{\lambda}(t),l(t))$ and the action $a(t)$, and is expressed by:
\begin{equation}
r_t=\left\{
\begin{array}{cl}
-f(l(t)),& \mbox{ if } a(t)=0,\\
\Phi- c_s-P_p-f(l(t)) & \mbox{ if } a(t)\geq 1 \mbox{ and } \theta(t)=0,\\
\Phi- c_s-P_{3G}-f(l(t)), & \mbox{ if } a(t)=2  \mbox{ and } \theta(t)=1.\\
- c_s-f(l(t)), & \mbox{ if } a(t)=1 \mbox{ and } \theta(t)=1.
\end{array}
\right.
\end{equation}

The problem faced by the SU consists of finding the policy $\mu$ that maximizes its expected average reward defined by:
$$\bar{R}(\mu)=\mathop {\lim }\limits_{T \to \infty}\frac{1}{T}\E_{\mu}\left(\sum_{t=1}^T r_t((\vec{\lambda}(t), l(t)),a(t))|{\vec{\lambda}(0)}\right),$$
\normalsize
where $\vec{\lambda}(0)$ is the initial belief vector. Thus, our objective is to find an optimal sensing policy $\mu^*$ that maximizes the average reward $\bar{R}(\mu)$, i.e.:
\begin{equation}\label{avpol}
    \mu^*=arg\max_{\mu \in \Gamma}\mathop {\lim }\limits_{T \to \infty}\frac{1}{T}\E_{\mu}\left(\sum_{t=1}^T r_t((\vec{\lambda}(t),l(t)),a(t))|{\vec{\lambda}(0)}\right).
\end{equation}

In some particular MDP and POMDP problems, we are able to determine an optimal policy in a smaller set reduced to stationary policies. 
Since we have a POMDP with a discrete state and action space, our POMDP framework can be transformed into a MDP problem over the belief state space \cite{Kaelbling98}. Then, the proof of the existence of an average optimal stationary policy results from Theorems  8.10.9 and 8.10.7 of \cite{Putterman}.

\begin{remark}
Let $\pi^{\mu^*}$ be the stationary distribution of the Markov chain $({\vec{\lambda}(t)},l(t))$ when SU uses the optimal stationary policy $\mu^*$.  Applying Little's result, the expected delay $E(D)$ is given by $E(D)=1+\frac{1}{thp}$, where $thp$ is the average throughput which is defined  as the expected number of departures per slot.  The throughput can be computed as follows
$$
thp= \sum_{\vec{\lambda}}   \pi^{\mu^*}(\vec{\lambda},1)
$$
Hence  a delay constraint may be implicitly controlled  by the penalty  $f(l)$.
\end{remark}

Given this result, we can restrict our problem to the set $\Gamma_S$ of stationary policies. Then, for the remainder of this paper, we omit the time index $t$ and we look for an optimal sensing policy which is a mapping between a system state $(\vec{\lambda},l)$ to an action $a$, independently of the time slot $t$. Now, we make a first analysis of the value function of the POMDP.

We denote by $\Omega^{ns}(\vec{\lambda}|\theta)$ the function that updates the belief vector $\vec{\lambda}$ when the user chooses to be inactive in the current slot, i.e. the SU takes action 0. The function $\Omega^{s}(\vec{\lambda}|\theta)$  updates the belief vector $\vec{\lambda}$ when the SU senses a licensed primary channel in the current slot and observes $\theta$, i.e. the SU takes the action 1 or 2.

The value function is denoted $V(\vec{\lambda},l)$. Let us denote by $Q_a(\vec{\lambda},l)$ the action-value function taking the action $a$ in the current slot when the information state is $(\vec{\lambda},l)$. Therefore, the value function is expressed by
\begin{equation}\label{vf}
g_u+V(\vec{\lambda},l)=\max_{a\in \mathcal{A}}Q_a(\vec{\lambda},l),
\end{equation}
where $g_u$ is a constant, and the optimal action is given by
\begin{equation}\label{opaction}
a^*(\vec{\lambda},l)=\arg\max_{a\in\mathcal{A}}Q_a(\vec{\lambda},l).
\end{equation}

We determine the action-value function for each different action 0, 1 and 2. When the SU decides to wait, i.e. to take the action $a=0$, we have:
\begin{eqnarray}
 Q_0(\vec{\lambda},l)=-f(l)+V(\Omega^{ns}(\vec{\lambda}|\theta=0),l+1).
\label{maxexp0}
\end{eqnarray}
When the SU chooses to sense the channel $n^*$ and decides to wait for the next time slot if the channel $n^*$ is busy ($a=1$), we have:
\begin{eqnarray}
 Q_{1}(\vec{\lambda},l)&=&-c_s+\lambda_{n^*}(\Phi-P_p+V(\Omega^{s}(\vec{\lambda}|\theta=0),1))\\ \nonumber &&+(1-\lambda_{n^*})(-f(l)+V(\Omega^{s}(\vec{\lambda}|\theta=1),l+1)).
\label{maxexp1}
\end{eqnarray}
\normalsize
When the SU chooses to sense the channel $n^*$ and to transmit using the dedicated channel if the channel $n^*$ is busy ($a=2$), we have:
\begin{eqnarray}
 Q_{2}(\vec{\lambda},l)&=&\Phi-c_s+\lambda_{n^*}(-P_p+V(\Omega^{s}(\vec{\lambda}|\theta=0),1))\\ \nonumber &&+(1-\lambda_{n^*})(-P_{3G}+V(\Omega^{s}(\vec{\lambda}|\theta=1),1)).
\label{maxexp2}
\end{eqnarray}
\normalsize

We take the assumption that there exists a packet delay $l^*$ such that the SU transmits its packet using the dedicated channel if the observation is $\theta=1$. In fact, this assumption is somehow realistic as the user has no interest to keep the packet in its buffer indefinitely.

We denote by $\alpha_n$ and $\beta_n$ the transition rates of the licensed primary channel $n$, and $\lambda_n$ the belief of the SU. We consider  that $\alpha_n\geq\beta_n$. When $\alpha_n\leq\beta_n$, the analysis is similar and the results are unchanged.
Let us focus on the belief update function $\Omega^{ns}$.
\begin{lemma}\label{beliefupdate}
We have the following properties of the belief update function $\Omega^{ns}$.
\begin{enumerate}
\item The update function $\Omega^{ns}(\lambda_n|\theta)$ is increasing with belief $\lambda_n$.
\item We have the following equivalence:
$$\Omega^{ns}(\lambda_n|\theta) \geq \lambda_n \quad \Leftrightarrow \quad \lambda_n\leq\pi_n(0),$$
and
$$\Omega^{ns}(\lambda_n|\theta) \leq \lambda_n \quad \Leftrightarrow \quad \lambda_n\geq\pi_n(0),$$
where $\pi_n(0)=\frac{\beta_n}{1-\alpha_n+\beta_n}$ is the stationary probability that the licensed primary channel $n$ is idle. Figure \ref{belevo} depicts the belief evolution depending on the packet delay.
\end{enumerate}
\begin{figure}[!h]
\centering
  \includegraphics[width=0.55\textwidth]{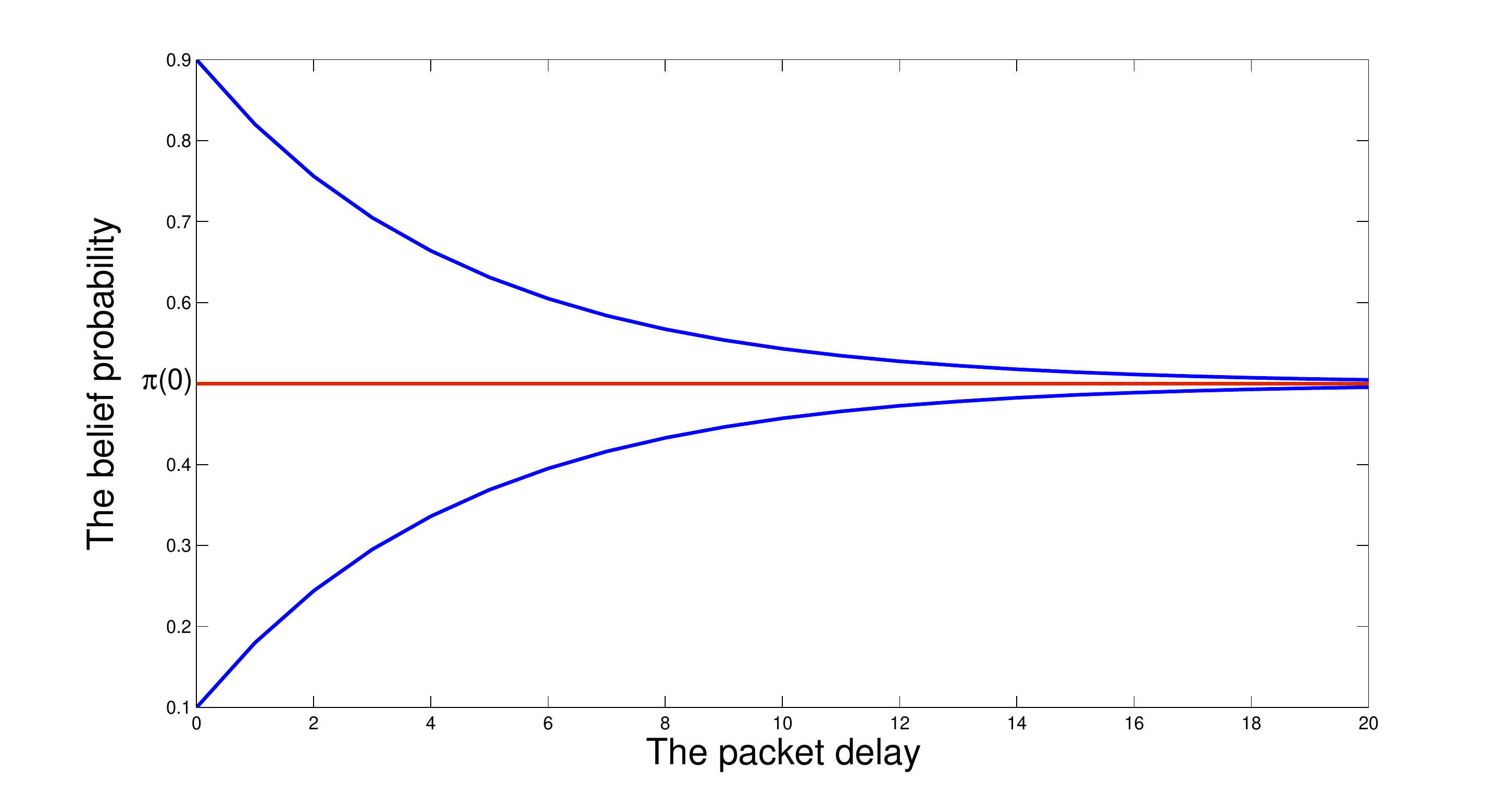}
  \caption{The belief update function $\Omega^{ns}$ with respect to the packet delay.}\label{belevo}
\end{figure}
\end{lemma}
\begin{proof}
See Appendix \ref{apbeliefupdate}.
\end{proof}

It has been shown in \cite{sondik} that the value function for a POMDP over a finite time horizon is piecewise linear and convex with respect to the belief vector. In Proposition \ref{pwlc}, we show that the value function for our POMDP problem over an infinite horizon with the average criterion, has also this property.
\begin{proposition}\label{pwlc}
The value function $V(\vec{\lambda},l)$  is piecewise linear and convex with respect to the belief vector $\vec{\lambda}$.
\end{proposition}
\begin{proof}
See Appendix \ref{appwlc}.
\end{proof}


Note that monotonicity results help us for establishing the structure of the optimal policies (see \cite{lovejoy} for an example) and provide insights into the underlying problem. The following propositions states monotonicity results of the value function with respect to the packet delay.
\begin{proposition}\label{monotonicityl}
For each belief vector $\vec{\lambda}$, the value function is monotonically decreasing with the packet delay $l$, i.e. $V(\vec{\lambda},l)\leq V(\vec{\lambda},l')$ for $l\geq l'$.
\end{proposition}
\begin{proof}
See Appendix \ref{apmonotonicityl}.
\end{proof}

This result is intuitive because for the same belief $\vec{\lambda}$ and for a given packet delay, the maximum expected remaining reward that can be accrued is lower than the one the SU can get with a smaller packet delay. 

\section{Optimal Threshold policy for single channel}\label{section4}
 The monotonicity with respect to the belief vector depends on the order relation over the belief set and also on the monotonicity of the belief update functions $\Omega^s(\vec{\lambda}|\theta=0)$ and $\Omega^s(\vec{\lambda}|\theta=1)$ depending on the belief vector. Thus, we can determine the structure of the optimal policy only for the single primary channel case.
\begin{proposition}\label{monotonicityb}
Denote $\lambda$ the belief probability of the licensed primary channel. The value function is monotonically increasing with the belief vector ${\lambda}$, i.e. $V({\lambda},l)\geq V({\lambda}',l)$ for ${\lambda}\geq{\lambda}'$. \end{proposition}
\begin{proof}
See Appendix \ref{apmonotonicityb}.
\end{proof}

{Again this result seems somehow intuitive as for the same packet delay, when the belief vector is higher the maximum expected remaining reward becomes higher.}

Given all the previous results on the value function $V({\lambda},l)$, we are able to show the existence of an optimal OSA policy for our POMDP problem. Moreover, we determine explicitly the threshold structure of such optimal policy.

Let us focus on the characteristics of an optimal policy for the SU. Intuitively, when the delay $l$ is small, the SU may choose to waits for a better opportunity. Thus, depending on the belief probability, the SU makes the decision to sense a primary channel or not. We prove in this section, that the intuition is true and there exists an optimal sensing policy which has a threshold structure.

The first decision for a SU is whether to sense the licensed primary channel or to wait, depending on its belief $\lambda$ and the current delay of the packet $l$. We have the following result which gives us a threshold on the belief probability in order to answer this question.
\begin{proposition}\label{thr}
For all packet delay $l$, the optimal action for the SU is to wait for the next slot, i.e. $a^*(\lambda,l)=0$ if and only if
$
\lambda \leq \lambda^*
$
where $\lambda^*$ is the solution of the equation $\lambda^*=\max(0,\min\{Th1(\lambda^*,l),Th2(\lambda^*,l)\})$ with
{\small{
\begin{eqnarray*}\nonumber
Th1(\lambda^*,l)=\frac{V(\Omega^{ns}(\lambda^*|\theta),l+1)-V(\beta,l+1)+c_s}{f(l)+\Phi-P_p+V(\alpha,1)-V(\beta,l+1)},
\mbox{and} \\ \nonumber Th2(\lambda^*,l)=\frac{V(\Omega^{ns}(\lambda^*|\theta),l+1)-V(\beta,1)+c_s-f(l)-\Phi+P_{3G}}{-P_p+V(\alpha,1)+P_{3G}-V(\beta,1)}.
\end{eqnarray*}}}
\normalsize
\end{proposition}

\begin{proof}
See Appendix \ref{apthr}.
\end{proof}


This proposition gives us a necessary and sufficient condition on the use of the action 0 depending on the belief probability $\lambda$. Consequently, if $\lambda > \lambda^*$ then the optimal action is to sense a primary channel, i.e. $a^*(\lambda,l)\neq 0$.

Furthermore, we have the following property of the optimal policy.
\begin{proposition}\label{aftersucc}
For all $\lambda>\pi(0)$ and $l$, the SU never takes the action $0$ and thus, $Q_0(\lambda,l)< \max{\{Q_{1}(\lambda,l),Q_{2}(\lambda,l)\}}$.
\end{proposition}
\begin{proof}
See Appendix \ref{apaftersucc}.
\end{proof}


Therefore, the SU never chooses the action $0$ after it transmits a packet over the primary channel because $\Omega^s({\lambda},\theta=0)=\alpha>\pi(0)$. Furthermore, we have the following result about the use of the dedicated channel.

\begin{proposition}\label{propl}
For all belief $\lambda$, the SU chooses to use the dedicated channel in spite of waiting for the next slot ($a^*(\lambda,l)=2$) if and only if the delay $l$ of the current packet verifies:
\begin{eqnarray*}
-f(l)-\Phi+P_{3G}+V(\beta,l+1)-V(\beta,1)>0.
\end{eqnarray*}
\normalsize
\end{proposition}
\begin{proof}
See Appendix \ref{appropl}.
\end{proof}


We note that this expression does not depend on the cost of sensing $c_s$ nor on the belief vector $\lambda$. That is obvious as this expression determines the best action to do after sensing a channel. We have the last property about the optimal threshold policy.

\begin{Corollary}[Never Wait After Sensing]\label{corol2}
If, for all $l$, the penalty cost $-f(l)$ is lower than $\Phi-P_{3G}$, then  the SU transmits on the dedicated channel when the sensed channel is not idle.
\end{Corollary}

\begin{proof}
See Appendix \ref{apcorol2}.

\end{proof}

This result is also somewhat intuitive. In fact, when the SU senses the channel as busy, it gets $\Phi-P_{3G}$ as reward if he uses the dedicated channel otherwise he gets a penalty $-f(l)$ if he decides to wait. Thus, if $\Phi-P_{3G}+f(l)$ is positive the SU has no incentive to wait after sensing the licensed primary channels.

 In the literature, the transition rates $\alpha$ and $\beta$ are assumed to be known by the SU. We focus in the next section on online learning algorithms that allow the SU to estimate those rates on the fly. In fact, in practice, some information like the transition rates $\alpha$ and $\beta$ are not available for the SU.

\section{Online Policy Learning}\label{seclearning}
\subsection{Online Learning of PU's activity}\label{learnPU}
In this section, we consider a model where the SU does not have external information about the state transition rates.
SU begins with an initial arbitrary values of $\alpha$ and $\beta$. He updates them every time slot depending on the information about the system state. Then, the SU computes its sensing policy based on the estimators $\hat{\alpha}=\{\hat{\alpha}_1,...,\hat{\alpha}_N\}$ and $\hat{\beta}=\{\hat{\beta}_1,...,\hat{\beta}_N\}$ where $\hat{\alpha}_i$ (resp. $\hat{\beta}_i$) is the estimator of $\alpha_i$ (resp. $\beta_i$).

First, the SU estimates $\hat{\alpha}_i$ which is the probability that the channel $i$ will be sensed idle given that it was idle in the previous slot. Second, the SU estimates $\hat{\pi}_i(0)$ the stationary probability for this channel to be idle. The SU obtains the estimated value of $\beta_i$ based on the relation $\hat{\beta}_i=(1-\hat{\alpha}_i)\frac{\hat{\pi}_i(0)}{1-\hat{\pi}_i(0)}$.

Formally, we consider the following counting processes for the estimation of $\hat{\alpha}_i$ and $\hat{\pi}_i(0)$:

\begin{itemize}
  \item The vector $\hat{K}=\{\hat{K}_1,...,\hat{K}_N\}$ where $\hat{K}_i$ represents the number of time slots a channel stays in the idle state, i.e. $\hat{K}_i$ is incremented if the channel $i$ is sensed and is idle at time slot $t$ and $t-1$.
  \item The vector $\hat{I}=\{\hat{I}_1,...,\hat{I}_N\}$ where $\hat{I}_i$ represents the number of time slots that the channel is sensed and is idle.
  \item The vector $\hat{M}=\{\hat{M}_1,...,\hat{M}_N\}$ where $\hat{M}_i$ represents the number of time slots that the channel is sensed.
\end{itemize}

Therefore the SU estimates the state transition rates $\hat{\alpha}$ and $\hat{\pi}_i(0)$ based on the following expressions:
$\hat{\alpha}_i=\frac{\hat{K}_i}{\hat{I}_i}$ and $\hat{\pi}_i(0)=\frac{\hat{I}_i}{\hat{M}_i}$.
\subsection{Learning Algorithm}
 Since solving POMDPs suffers from the higher computational complexity, we consider that the SU do not solve the POMDP defined in Section \ref{section3}. Instead, we suppose that the SU has two options:
\begin{itemize}
\item The SU sends the channel transitions to a server in which the POMDP problem is solved offline for different values of channel transitions.
\item Knowing that the optimal OSA policy has a threshold structure, the SU computes an optimal OSA policy using an online learning algorithm.
\end{itemize}

We focus, in this section, on the second option and we propose an online learning algorithm that allow the SU to determine the OSA policy on the fly. We propose an on-policy Sarsa-based learning algorithm, where the SU maintains a state-action Q-value $Q(\alpha, \beta, \Lambda^*)$. For each value of transition rate, estimated by the SU, the SU chooses the threshold policy that maximizes its state-action Q-value: $\Lambda^*=\arg\max\limits_{\Lambda}Q(\alpha, \beta, \Lambda)$. Note that $\Lambda^*=\{\lambda^*_1,\lambda^*_2,\cdots\}$, where $\lambda^*_i$ is the threshold belief probability below which the SU do not sense licensed primary channels when the delay of its packet equals $i$. In Algorithm \ref{algo1}, we have used an aggregation parameters $m$ in order to transform the continuous space of channel transitions into a discrete one. In fact, we consider that $\alpha_i=\frac{k}{m}$ if $\alpha_i\in[\frac{k}{m},\frac{k+1}{m}], 0\leq k\leq m$. Indeed, increasing $m$ increases the accuracy of the algorithm, however it increases also the memory requirements. Once the SU estimates the channels transitions, it chooses a threshold policy that it can not change before $nbsolt$ time slot.  $\rho_k$ is the learning rate factor satisfying $\sum_{k=1}^\infty\rho_k=\infty,\sum_{k=1}^\infty(\rho_k)^2<\infty$, e.g. $\rho_k=\frac{1}{k}$, and $\eta$ is the discount factor.

\begin{algorithm}
\caption{Learning-based algorithm for the SU}
\label{algo1}
\begin{algorithmic}
\STATE Initialize $Q(\alpha, \beta, \Lambda)=0$ for all channels transitions and threshold policies;
\STATE Initialize $\Lambda^*$  to a random value;
\STATE Set $R=0$;
\WHILE{true}
\STATE $\Lambda^*_{prev}=\Lambda$;
\STATE $\alpha^{prev}=\alpha$;
\STATE $\beta^{prev}=\beta$;
\STATE Estimate the channels transitions $\alpha$ and $\beta$ using the method described in Section \ref{learnPU};
\STATE Select the threshold policy $\Lambda^*$ as follows: $\Lambda^*=\arg\max\limits_{\Lambda}Q(\alpha, \beta, \Lambda)$ with probability $(1-\epsilon)$, else choose a random policy;
\FOR{$n=1\rightarrow nbslot$}
\STATE Transmit packet using the threshold policy $\Lambda^*$.
\STATE $R=R+r_t((\vec{\lambda}, l),a)$;
\ENDFOR
\STATE $Q(\alpha^{prev},\beta^{prev},\Lambda^*_{prev})\leftarrow \rho_kQ(\alpha^{prev},\beta^{prev},\Lambda^*_{prev}) +(1-\rho_k)(R+\eta Q(\alpha,\beta,\Lambda^*))$;
\STATE $R=0$;
\STATE $k=k+1$;
\ENDWHILE
\end{algorithmic}
\end{algorithm}


\section{Numerical Illustrations}\label{section5}

In this section, we validate our results through simulations of the system over an important number of packets (we consider 3000 packets).  We consider the following system parameters: $P_{3G}=800$, $P_p=100$, $c_S=50$ and $\Phi=350$ bits. We consider the delay penalty function $f(l)=\gamma \log(l)$, where $\gamma$ is the delay penalty parameter.  {We investigate the optimal policy for the SU, and its threshold structure, in the single channel model and in the multi-channel model. Moreover, we show how we can tune the system parameters (delay penalty and sensing cost) in order to obtain a target packet' delay or energy consumption. Thereafter, we compare our proposed threshold-based OSA policy with a set of memoryless policies. Finally, we illustrate how the SU learns the PUs' activity and the OSA policy on the fly.}

 \subsection{Multiple channel model}

  We consider the following three scenarios with symmetric channels:
\begin{enumerate}
  \item Scenario 1: Licensed primary channels are often occupied ($\alpha_1=\alpha_2=\alpha_3=\alpha_4=0.15$ and $\beta_1=\beta_2=\beta_3=\beta_4=0.1$),
  \item Scenario 2: Licensed primary channels are often idle ($\alpha_1=\alpha_2=\alpha_3=\alpha_4=0.85$ and $\beta_1=\beta_2=\beta_3=\beta_4=0.7$),
  \item Scenario 3: Licensed primary channels have low transition rates ($\alpha_1=\alpha_2=\alpha_3=\alpha_4=0.95$ and $\beta_1=\beta_2=\beta_3=\beta_4=0.05$). This last scenario is realistic if we consider TV white space \cite{SS09}.
\end{enumerate}
We consider 4 i.i.d licensed primary channels, i.e. $N=4$, due to exponential states space and we set $\gamma=10$. We simulate the three scenarios and we depict in Figure \ref{figoccup} the thresholds $\lambda^*(l)$ determined in proposition \ref{thr} depending on the packet delay $l$ for each scenario.We observe that the SU policy has also a threshold structure. For every packet delay $l$, the best action for the SU is to wait for the next slot if its belief probability is lower than $\lambda^*$. Otherwise, he senses a licensed primary channel. In this context, where licensed primary channels are often occupied (Scenario 1, Figure \ref{figoccup}), the maximum packet delay $l^*$ obtained with Proposition \ref{propl} equals 9.  The maximum packet delay for scenarios 2 and 3 is $l^*=5$.  Note that the  threshold belief probability $\lambda^*$ is not decreasing with the packet delay.  In fact, since licensed primary channels are more static (the probability for each channel to stay occupied or idle is high enough), it appears one kind of periodic threshold strategy.

The sensing probability presented on the y-axis in Fig. \ref{figoccup},\ref{learningpng} and  \ref{1ch1} refer to the belief probability introduced in Section III, Equation (1). At each time slot t, based on its belief vector ${\vec{\lambda}(t)}$, the SU chooses a channel to be sensed. There exists several channel choice policies in the literature such as deterministic, randomized and periodic. In this paper, we consider that the SU senses the channel, which has the highest probability to be idle. 

\begin{figure}
  \begin{center}
      \includegraphics[width=0.5\textwidth]{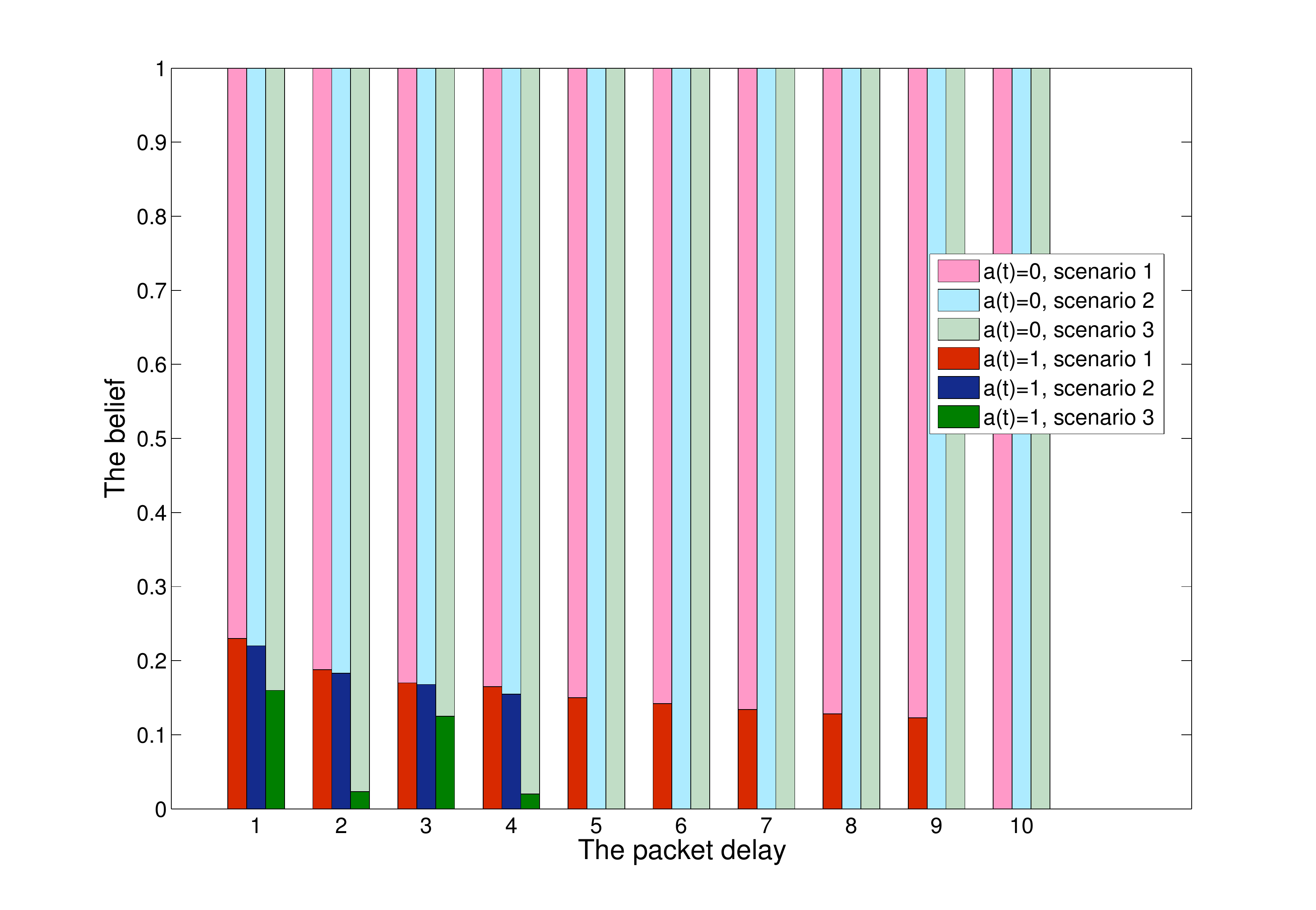}
      \caption{Optimal policy for the SU in the multichannel case for scenarios 1, 2 and 3.}
      \label{figoccup}
  \end{center}
\end{figure}

\subsection{Online policy learning}
We consider 4 i.i.d licensed primary channels, i.e. $N=4$, and we simulate the first scenario. We depict, in Figure \ref{learningpng}, the OSA learning obtained after $200$ iterations of the learning algorithm proposed in Section \ref{seclearning}. Note that even if the learning algorithm gives a suboptimal OSA policy, it allows the SU to determine a near optimal OSA policy on the fly. We observe also that the learning algorithm leads to a less risky policy compared to the optimal one, in the sense that a any packet delay $l$ the sensing probability is higher with the learning compared with the one obtained with the optimal policy.
\begin{figure}[!h]
  \includegraphics[width=0.5\textwidth]{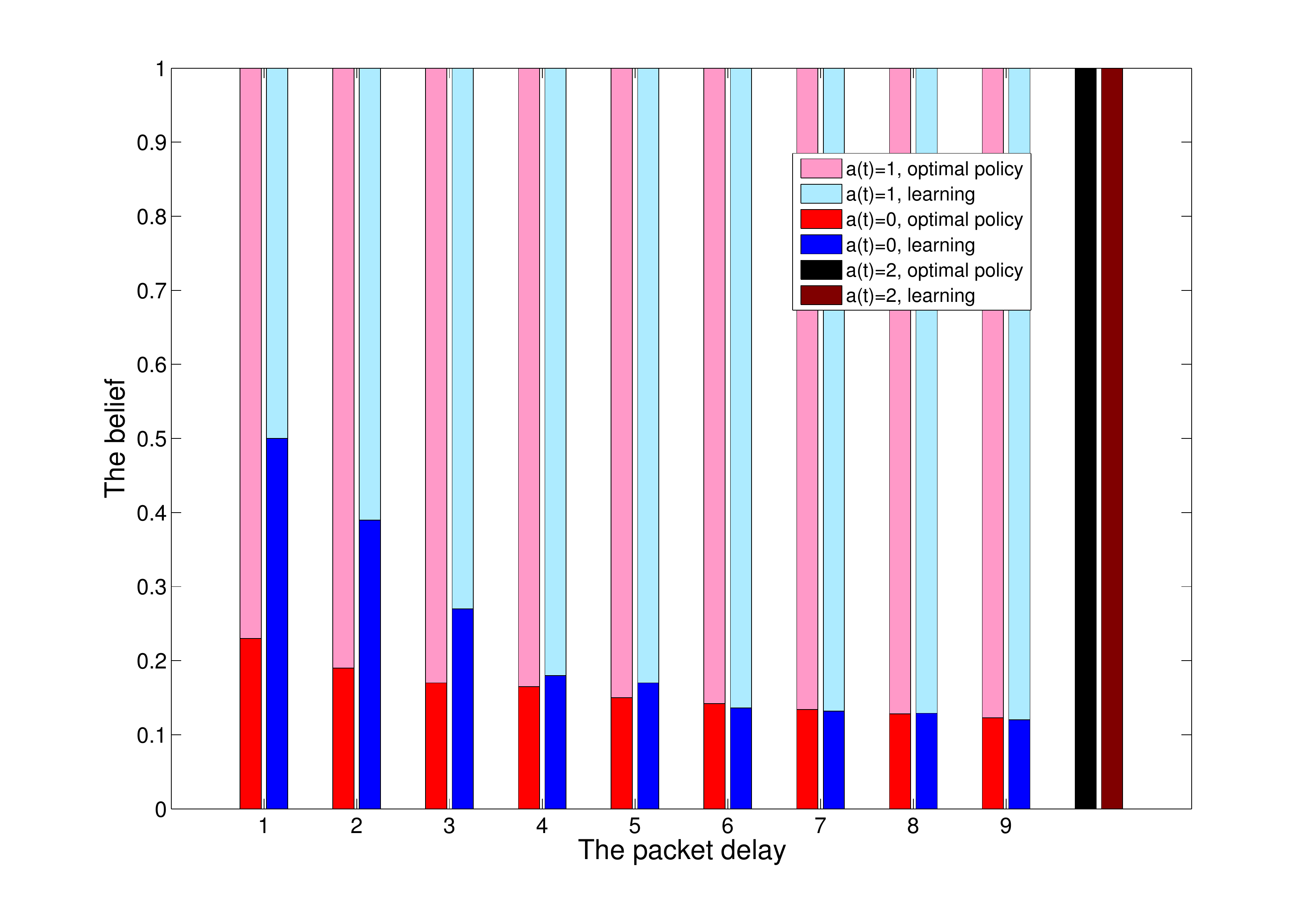}
  \caption{OSA policy for the SU with online learning for scenario 1}\label{learningpng}
\end{figure}

\subsection{Single channel model}
\subsubsection{Impact of the sensing cost}
Let us consider a SU and one channel licensed to PU. We simulate  a scenario where the transition rates $\alpha=0.15$ and $\beta=0.1$. We illustrate, in this section, the impact of the sensing cost on the optimal OSA policy of the SU. Figure \ref{1ch1} depicts the optimal policy of the SU depending on the belief and the packet delay, for different values of sensing cost ($c_s=50$ and $c_s=200$). For each packet delay, the SU has a threshold policy depending the belief probability. Indeed, given the packet delay, if the belief probability of the SU is higher than the threshold he senses the licensed primary channel, otherwise he remains idle and waits for the next time slot. Specifically, even if we were not able to prove analytically that the belief threshold is decreasing with respect to the packet delay, we observe, in Figure \ref{1ch1}, that the threshold belief probability $\lambda^*$ is decreasing with packets' delays in both scenarios. Note that the SU waits for the next time slot if the channel is sensed as busy until the packet delay equals 13 for $c_s=50$ (and 3 for $c_s=200$), then he transmits the packet using the dedicated channel.  Indeed, as the sensing cost increases, the SU has less incentive to sense licensed primary channels.
\begin{figure}
  \begin{center}
      \includegraphics[width=0.5\textwidth]{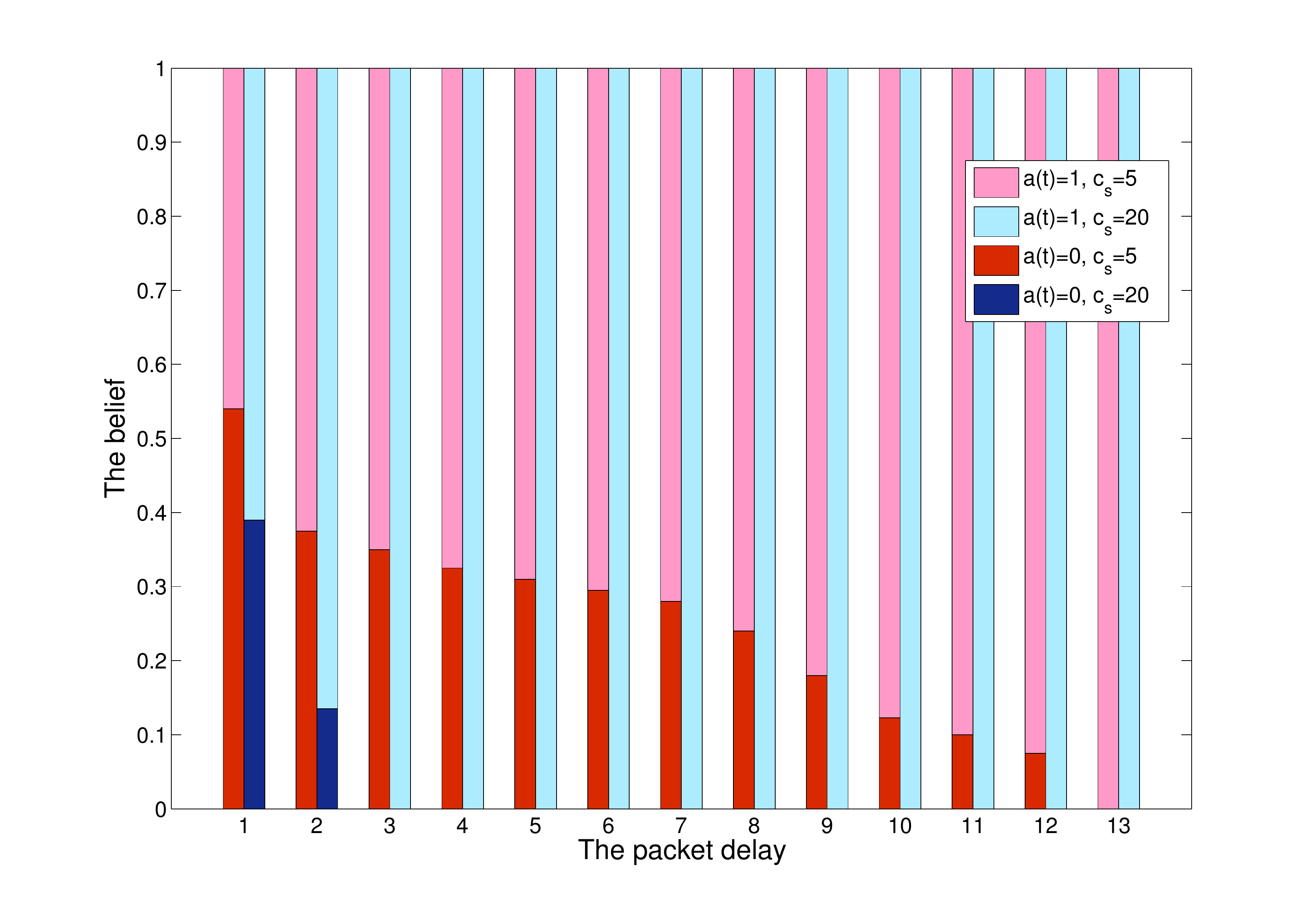}
      \caption{Optimal policy with one licensed primary channel for $c_s=5$ and $c_s=20$, $\alpha=0.15$ and $\beta=0.1$.}
      \label{1ch1}
  \end{center}
\end{figure}

\subsubsection{Impact of the delay penalty}
We investigate, in this section, the impact of the delay penalty on performance metrics like the average packet delay and the average energy consumption per packet, using the optimal policy. Indeed, it is possible to tune the delay penalty parameter $\gamma$ in order to obtain targeted values for the average delay and for the energy consumption. We illustrate, in Figure \ref{del_op}, the average delay, obtained with the optimal policy, as a function of the penalty parameter $\gamma$. In fact, we observe that the average delay is strictly decreasing with the delay penalty. This result is somehow intuitive as the user has less incentive to wait for next time slots when the penalty of the delay increases. Moreover, we plot in Figure \ref{cos_op} the average energy consumption per time slot  depending on the delay penalty $\gamma$. Indeed, the higher is the penalty $\gamma$, the lower is the average delay and the higher is the energy consumption, since the SU transmits more often over the dedicated channel. In fact, Figure \ref{cos_op} show that the energy consumption curve is S-shaped where the consumed energy increase quickly for lower values of $\gamma$ and tends to be unchanged for higher values $\gamma$.
\begin{figure}[!h]
  \includegraphics[width=0.55\textwidth]{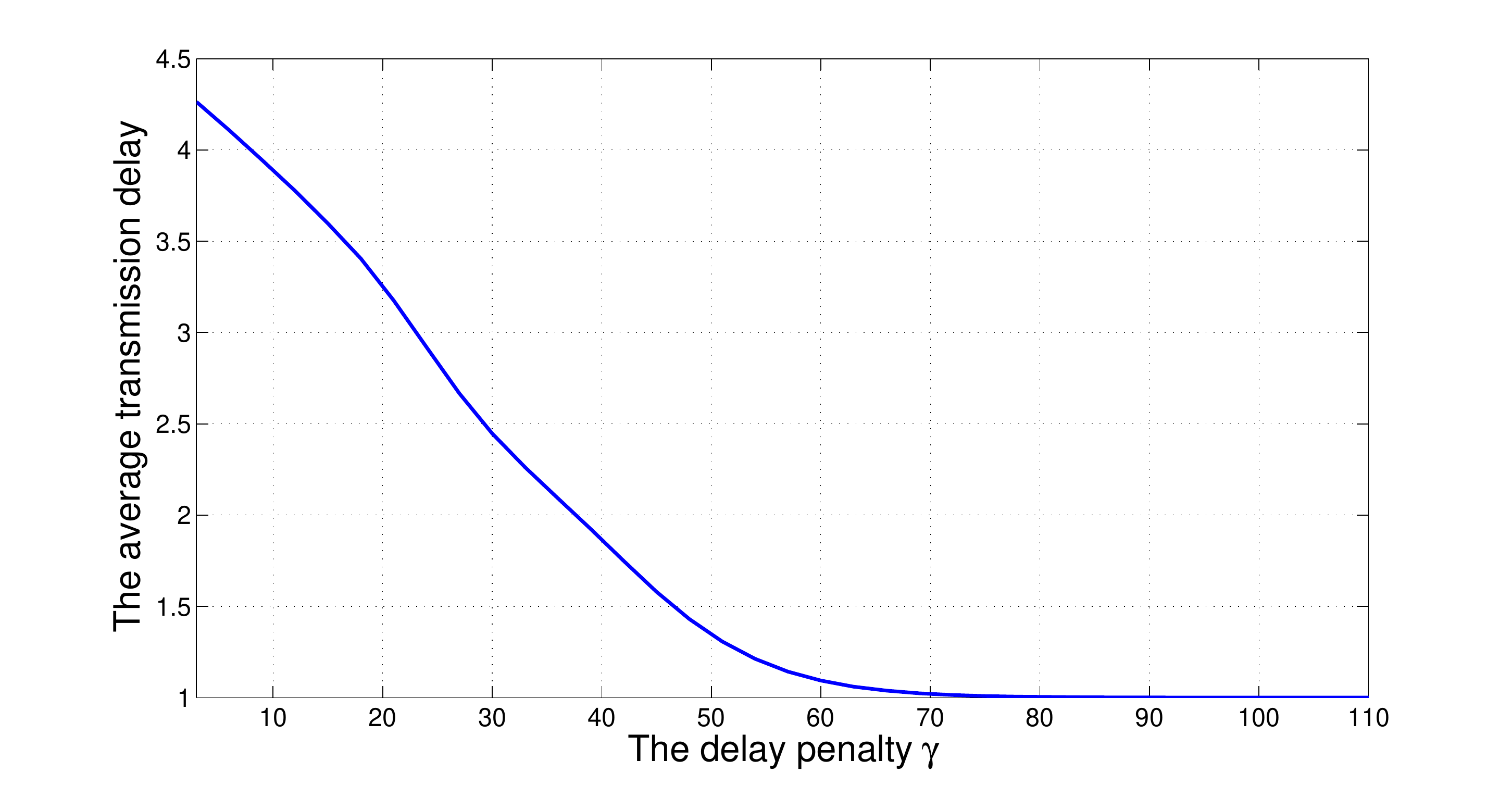}
  \caption{The average packet delay depending on the delay penalty $\gamma$, $\alpha=0.15$ and $\beta=0.1$}\label{del_op}
\end{figure}

\begin{figure}[!h]
  \includegraphics[width=0.55\textwidth]{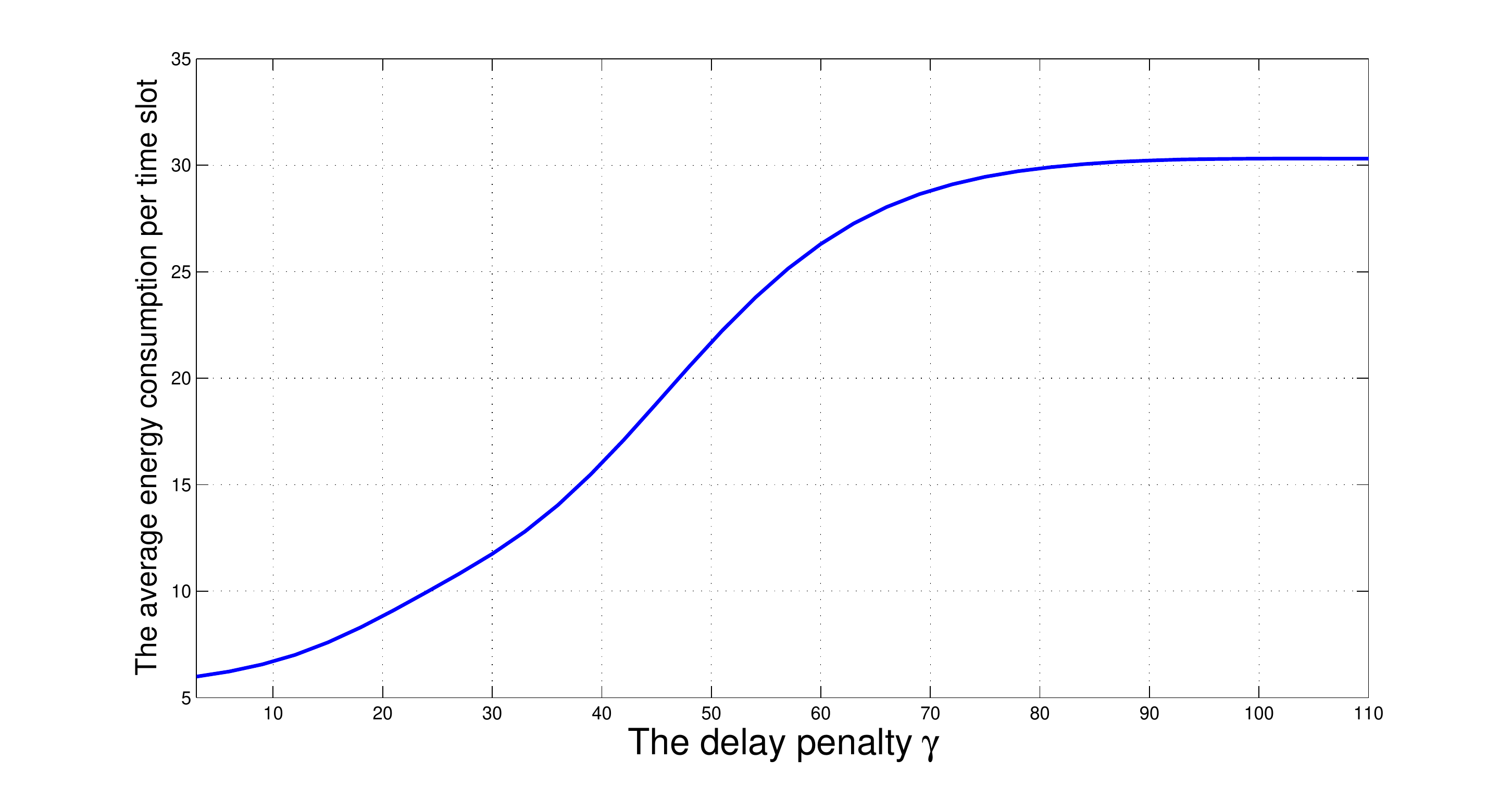}
  \caption{The average energy consumption per time slot depending on the delay penalty $\gamma$, $\alpha=0.15$ and $\beta=0.1$}\label{cos_op}
\end{figure}

\subsubsection{Optimal policy vs Memoryless policies}

We compare the performance of the optimal policy obtained by Algorithm \ref{algo1} with memoryless policies (MP) which are defined as follows: The SU senses and  transmits if the channel is idle. A memoryless policy is characterized by the number of attempts (always finding an occupied channel) before using the dedicated channel. For example, using the memoryless policy denoted (MP-3),  the SU senses the channel and transmits if the channel if idle, otherwise he waits for the next time slot until the packet delay equals to 3, then he transmits using the dedicated channel if the unlicensed channel is occupied. For each memoryless policy, we determine the average delay and the average energy consumption per packet. Note that we are considering several MP because every MP allows SUs to obtain a given QoS, and we are trying to evaluate the performance of our proposed policy for different values of the QoS.
In fact, our goal is to illustrate the gain of energy consumption using our optimal policy compared to  memoryless policies, when considering the same quality-of-service, i.e. average delay here. Thus, we tune the delay penalty $\gamma$ such that our optimal policy has the same average delay as the MP. The percentage of the average cost reduction (sensing cost and transmission cost) per packet when using the optimal threshold policy is compared to MP in Figure \ref{del_diff}, for different values of the average delays. We observe that for our proposed policy the average cost reduction is higher compared to MP (up to 50\%). Indeed, our policy is well adapted for applications that require hard transmission delays.
\begin{figure}[!h]
  \includegraphics[width=0.55\textwidth]{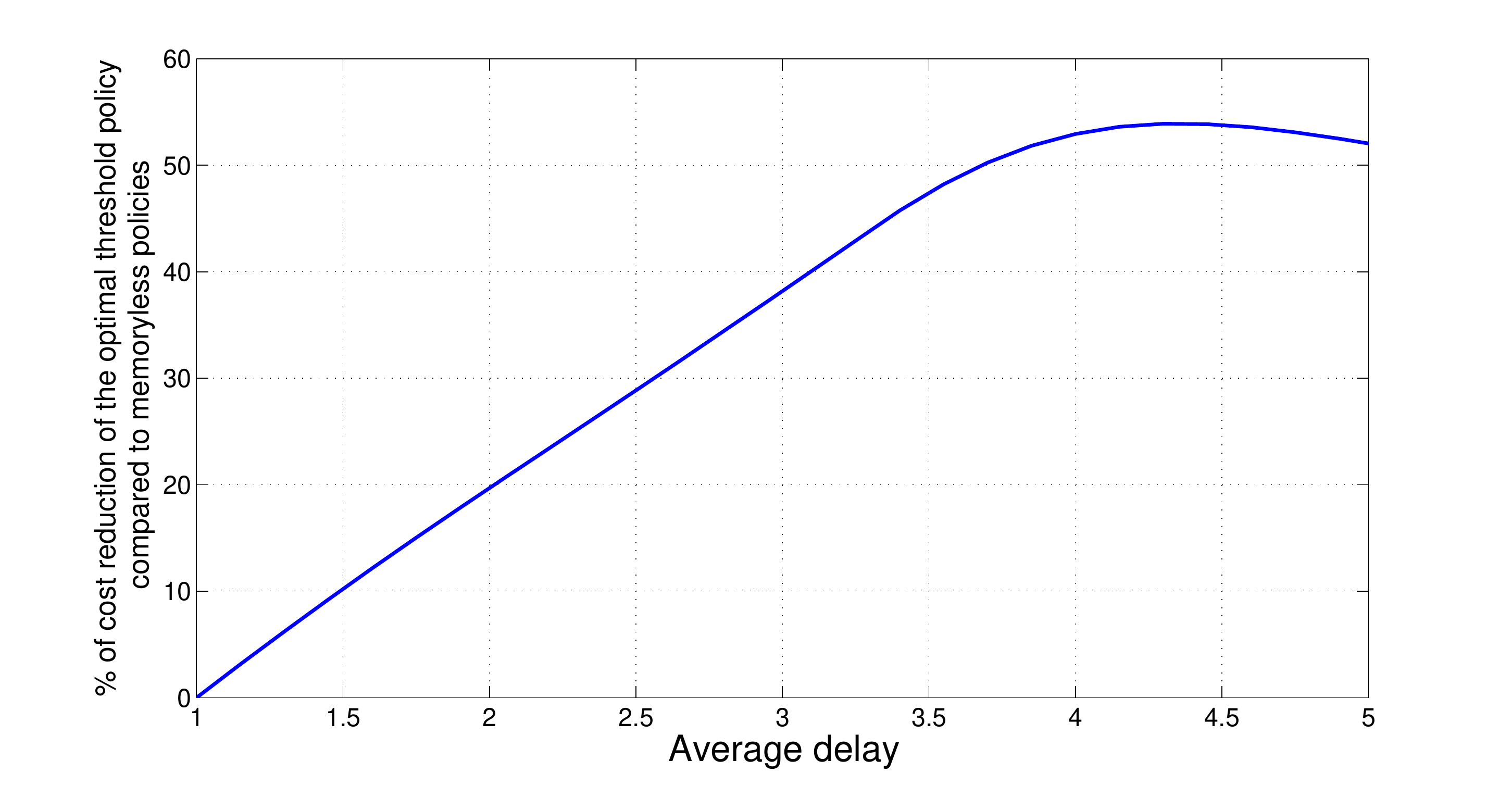}
  \caption{The difference of energy consumption between the optimal policy and the memoryless policy depending on the average delay.}\label{del_diff}
\end{figure}

\section{Conclusion and perspectives}
In this paper, we have used a POMDP framework for determining an optimal OSA policy  taking into account an energy-delay tradeoff for SUs. Introducing a QoS metric in the spectrum sensing policy is very important with the emergence of heterogeneous mobiles that are able to transmit their traffic with possible high QoS, at any time over different ways of communication like 3G, WiFi and  TV White Space. We have provided some structural properties of the value function and then proved the existence of an optimal average stationary OSA policy. We have been able to determine explicitly the threshold structure of the optimal policy. Moreover, we have proposed a learning mechanism that determine the OSA policy on the fly.
There exists several OSA mechanisms in the literature and it is important for the community to design a generic framework in order to compare all existing approaches for OSA in cognitive radio networks, this is part of our future works. Furthermore,
we have considered a perfect sensing model where the SU senses the channel in a way to ensure that the PU is present or not. Mis-detection can be also integrated to our framework. Finally, the interaction between several SUs has not been considered here, and in the literature very few. This perspective is also very important because if the channel choice policy is the same for all the SUs, there could have lots of collisions between several SUs that have sensed the same idle primary channel. This decentralized system with partial information can be modeled using decentralized-POMDP or interactive-POMDP and will be studied in future works.

\appendix
\subsection{Proof of Lemma \ref{beliefupdate}}\label{apbeliefupdate}
First, the update function $\Omega^{ns}$ is linear with the belief because because $\Omega^{ns}(\lambda)=\beta+(\alpha-\beta)\lambda$. As we considered the case where $\alpha\geq\beta$, then the update function is increasing with the belief.

Second, let us prove that $\Omega^{ns}(\lambda)\geq\lambda$ if $\lambda\leq\pi(0)$ by induction on the belief.
\begin{enumerate}
\item We have the initial condition: $\beta \leq \pi(0)=\frac{\beta}{1-\alpha+\beta}$ and $\Omega^{ns}(\beta)=\beta+(\alpha-\beta)\beta\geq\beta$.
\item We assume that $\Omega^{ns}(\lambda)\geq\lambda$ for a given $\lambda\leq\pi(0)$.
\item The induction operator gives: $\Omega^{ns}(\Omega^{ns}(\lambda))=\beta+(\alpha-\beta)\Omega^{ns}(\lambda)\geq \beta+(\alpha-\beta)\lambda=\Omega^{ns}(\lambda)$.
\end{enumerate}
Thus, $\Omega^{ns}(\lambda)\geq\lambda$ for all $\lambda\leq\pi(0)$. The analysis for $\lambda\geq\pi(0)$ is similar.
\subsection{Proof of Proposition \ref{pwlc}}\label{appwlc}
The proof of the proposition \ref{pwlc} is similar to \cite{sondik} where the authors consider the finite time horizon problem. Hence, we briefly describe the procedure for this proof. Considering the maximum packet delay $l^*$ and for all belief vector $\lambda$, the value function $V(\vec{\lambda},l^*)$ is linear with the belief because

   {\scriptsize{
\begin{eqnarray*}
V(\vec{\lambda},l^*)&=&Q_2(\vec{\lambda},l^*)-g_u,\\
&=&-g_u+\Phi-c_s-P_{3G}+V(\Omega^{s}(\vec{\lambda}|\theta=1),1)+\\
\nonumber &&\lambda_{n^*}(P_{3G}-P_p+V(\Omega^{s}(\vec{\lambda}|\theta=0),1)-V(\Omega^{s}(\vec{\lambda}|\theta=1),1)).
\end{eqnarray*}}}
Then the value function $V(\vec{\lambda},l^*)$ can be rewritten as an inner product of the belief vector and a $\Upsilon$-vector. As $Q_2(\vec{\lambda},l)=Q_2(\vec{\lambda},l^*)$, for all $l$, the action-value function $Q_2(\vec{\lambda},l)$ can be also rewritten as an inner product of the belief vector and a $\Upsilon$-vector. We suppose that Proposition \ref{pwlc} holds for all packet delays higher than $l+1$ and we prove that the proposition is true for packet delay $l$. After some algebra, we can rewrite the action-value functions given in (\ref{maxexp0}) and (\ref{maxexp1}) in terms of $\Upsilon$-vector:
\begin{eqnarray}
\nonumber Q_0(\vec{\lambda},l)&=&-f(l)+\max\limits_{\Upsilon\in\Gamma_{l+1}}<\Omega^{ns}(\vec{\lambda}|\theta),\Upsilon>\\&=& -f(l)+\sum\limits_{s\in\mathcal{S}}\omega_{s}\left[\sum\limits_{s'\in\mathcal{S}}P(s'|s)\Upsilon_{l+1}^{\Omega^{ns}(\vec{\lambda}|\theta)}\right],
\label{q0}
\end{eqnarray}
and
{\small{
\begin{eqnarray}
\nonumber Q_1(\vec{\lambda},l)&=&-c_s+\lambda_{n^*}(\phi-P_p+V({\Omega^{s}(\vec{\lambda}|\theta=0)},1))+(1-\lambda_{n^*})\\\nonumber&&(-f(l)+\max\limits_{\Upsilon\in\Gamma_{l+1}}<{\Omega^{s}(\vec{\lambda}|\theta=1)},\Upsilon>)\\\nonumber
&=&-c_s+\lambda_{n^*}(\phi-P_p+V(\Omega^{s}(\vec{\lambda}|\theta=0),1))+(1-\lambda_{n^*})\\\nonumber&&(-f(l)+\sum\limits_{s\in\mathcal{S}}\omega_{s}\left[\sum\limits_{s'\in\mathcal{S}}P(s'|s)\Upsilon_{l+1}^{\Omega^{s}(\vec{\lambda}|\theta=1)}\right]),
\label{q1}
\end{eqnarray}}}
where $\Upsilon_{l+1}^{\Omega^{ns}(\vec{\lambda}|\theta)}$ and $\Upsilon_{l+1}^{\Omega^{s}(\vec{\lambda}|\theta=1)}$ are, respectively, the $\Upsilon$-vectors for the regions containing belief vectors $\Omega^{ns}(\vec{\lambda}|\theta)$ and $\Omega^{s}(\vec{\lambda}|\theta=1)$, respectively. Each term in the square brackets of (\ref{q0}) and (\ref{q1}) are elements $\Upsilon_{\lambda,l}$ of  a $\Upsilon$-vector $\Upsilon_{l}$. Then the action-value functions can be rewritten as an inner product of the belief vector and a $\Upsilon$-vector $\Upsilon_{l}$. Moreover, there are only a finite number of such $\Upsilon$-vector  $\Upsilon_{l}$ since we have a finite set of belief for all $l$. As the maximum of a finite set of piecewise linear and convex functions is also piecewise linear and convex, the Proposition \ref{pwlc} holds.
\subsection{Proof of Proposition \ref{monotonicityl}}\label{apmonotonicityl}
Let us prove first that the value function $V(\vec{\lambda},l)$ is monotonically decreasing with the packet delay $l$ for all belief vector $\vec{\lambda}$. The SU takes the action $2$ for all $\vec{\lambda}$ when the packet delay is $l^*$, thus we have:
\begin{eqnarray*}
V(\vec{\lambda},l^*)&=&\Phi-c_s+\lambda_{n^*}(-P_p+V(\Omega^{s}(\vec{\lambda}|\theta=0),1))+(1-\lambda_{n^*})\\&&(-P_{3G}+V(\Omega^{s}(\vec{\lambda}|\theta=1),1)).
\end{eqnarray*}
The SU chooses the action that maximizes its average utility and thus:

{\small{
\begin{eqnarray*}
V(\vec{\lambda},l^*-1)&=&\\\max_{a}Q_a(\vec{\lambda},l^*-1)-g_u&\geq& Q_2(\vec{\lambda},l^*-1)-g_u,\\
&=&\Phi-c_s+\lambda_{n^*}(-P_p+V(\Omega^{s}(\vec{\lambda}|\theta=0),1))\\&+&(1-\lambda_{n^*})(-P_{3G}+V(\Omega^{s}(\vec{\lambda}|\theta=1),1))-g_u,\\
&=&V(\vec{\lambda},l^*).
\end{eqnarray*}}}

Let us prove that this propriety holds for all packet delays using a backward induction on $l$:
\begin{enumerate}
\item initial condition: For all belief vector $\lambda$, $V(\vec{\lambda},l^*)\leq V(\vec{\lambda},l^*-1)$,
\item we suppose that $V(\vec{\lambda},l+2)\leq V(\vec{\lambda},l+1)$, $\forall \vec{\lambda}$.
\item We have:
{\small{
\begin{eqnarray*}
\nonumber Q_0(\vec{\lambda},l)&=&-f(l)+V(\Omega^{ns}(\vec{\lambda}|\theta),l+1),\\
\nonumber&\geq&-f(l+1)+V(\Omega^{ns}(\vec{\lambda}|\theta),l+2),\\
\nonumber &=&Q_0(\vec{\lambda},l+1).\\
\nonumber Q_1(\vec{\lambda},l)&=&-c_s+\lambda_{n^*}(\Phi-P_p+V(\Omega^{s}(\vec{\lambda}|\theta=0),1))\\\nonumber&&+(1-\lambda_{n^*})(-f(l)+V(\Omega^{s}(\vec{\lambda}|\theta=1),l+1)),\\
\nonumber&\geq&-c_s+\lambda_{n^*}(\Phi-P_p+V(\Omega^{s}(\vec{\lambda}|\theta=0),1))\\\nonumber&&+(1-\lambda_{n^*})(-f(l+1)+V(\Omega^{s}(\vec{\lambda}|\theta=1),l+2)),\\
\nonumber &=&Q_1(\vec{\lambda},l+1).\\
\nonumber Q_2(\vec{\lambda},l)&=&-c_s+\Phi-P_{3G}+V(\Omega^{s}(\vec{\lambda}|\theta=1),1)\\\nonumber&&+\lambda_{n^*}(P_{3G}-P_p+V(\Omega^{s}(\vec{\lambda}|\theta=0),1)-V(\Omega^{s}(\vec{\lambda}|\theta=1),1)),\\
\nonumber &\geq&Q_2(\vec{\lambda},l+1).
\end{eqnarray*}}}
The inequalities come from the induction assumption and the monotonicity of the penalty function $f(l)$. Thus, we have:
$\forall \lambda, \quad V(\lambda,l)\geq V(\lambda,l+1)$.
\end{enumerate}
The value function is therefore decreasing with the packet delay.

\begin{lemma}\label{cogvs3g}
We have the following inequality: $$-P_p+V(\alpha,1)\geq-P_{3G}+V(\beta,1).$$
\end{lemma}

\subsection*{Proof of Lemma \ref{cogvs3g}}\label{apcogvs3g}

We prove this lemma by contradiction, so we suppose that $-P_p+V(\alpha,1)<-P_{3G}+V(\beta,1)$.
We first prove that the following:
{\scriptsize{
\begin{eqnarray}
\nonumber g_u+V(\alpha,1)&\geq& Q_2(\alpha,1),\\
\nonumber g_u+V(\alpha,1)&\geq& -c_s+\alpha(\phi-P_p+V(\alpha,1))+(1-\alpha)(\phi-P_{3G}+V(\beta,1)),\\
\nonumber g_u+V(\alpha,1)&\geq& -c_s+\phi-P_p+V(\alpha,1),\\
\nonumber g_u&>&\Phi-c_s-P_p.
\end{eqnarray}}}
and we take the assumption that the immediate reward when the channel is idle is positive, i.e. $\Phi-c_s-P_p\geq0$.
We know that the SU takes the action 2 in the state  $(\lambda,l^*)$ for all belief vector $\lambda$, i.e $a^*(\lambda,l^*)=2, \forall \lambda$. We have:
{\scriptsize{
\begin{eqnarray*}
\nonumber g_u+V(\lambda,l^*)&=&-c_s+\lambda(\phi-P_p+V(\alpha,1))+(1-\lambda)(\phi-P_{3G}+V(\beta,1)).
\end{eqnarray*}}}
Let us focus on the packet delay $l^*-1$. If $\lambda\leq\pi(0)$,  we have:
{\small{
\begin{eqnarray*}
\nonumber Q_0(\lambda,l^*-1)&=&-f(l^*-1)+V(\Omega^{ns}(\lambda),l^*),\\
\nonumber &=&-g_u-f(l^*-1)-c_s+\Omega^{ns}(\lambda)(\phi-P_p+V(\alpha,1))\\\nonumber&+&(1-\Omega^{ns}(\lambda))(\phi-P_{3G}+V(\beta,1)),\\
\nonumber &=&V(\lambda,l^*)-f(l^*-1)+(\Omega^{ns}(\lambda)-\lambda)(P_{3G}-P_p\\\nonumber&&+V(\alpha,1)-V(\beta,1)),\\
\nonumber &<&V(\lambda,l^*).
\end{eqnarray*}}}
The inequality is due to the assumption that $-P_p+V(\alpha,1)<-P_{3G}+V(\beta,1)$, $\Omega^{ns}(\lambda)\geq\lambda$ and $f(l^*-1)$ is positive. As the value function $V(\lambda,l)$ is decreasing with the packet delay $l$ (see Proposition \ref{monotonicityl}), then $Q_0(\lambda,l^*-1)<V(\lambda,l^*)<V(\lambda,l^*-1)$. As we proved that $g_u\geq0$, the SU does not take the action $0$ when the packet delay is $l^*-1$. For the action 1, we have:
{\small{
\begin{eqnarray*}
\nonumber Q_1(\lambda,l^*-1)&=&-c_s+\lambda(\phi-P_p+V(\alpha,1))\\\nonumber&&+(1-\lambda)(-f(l^*-1)+V(\beta,l^*)),\\
\nonumber &=&-c_s+\lambda(\phi-P_p+V(\alpha,1))+(1-\lambda)\\\nonumber&&\left(\right.\phi-g_u-f(l^*-1)-c_s +\beta(-P_p+V(\alpha,1))\\
\nonumber &&\left.+(1-\beta)(-P_{3G}+V(\beta,1))\right),\\
\nonumber &<&-c_s+\lambda(\phi-P_p+V(\alpha,1))+(1-\lambda)\\\nonumber&&(\phi-g_u-f(l^*-1)-c_s-P_{3G}+V(\beta,1)),\\
\nonumber &<&-c_s+\lambda(\phi-P_p+V(\alpha,1))\\\nonumber&&+(1-\lambda)(\phi-P_{3G}+V(\beta,1)),\\
\nonumber &=&Q_2(\lambda,l^*-1).
\end{eqnarray*}}}
The first inequality is due to the assumption that $-P_p+V(\alpha,1)<-P_{3G}+V(\beta,1)$ and the second one is because $g_u$, $f(l^*-1)$ and $c_s$ are positive. Thus, the optimal strategy is to take the action 2 when the packet delay is $l^*-1$.


Let us prove now by backward induction on $l$ that the optimal action is the action 2 for all belief vector $\lambda \leq \pi(0)$.
\begin{itemize}
\item If the SU takes the action 2 when the packet delay is $l^*$, then it takes also the action 2 when the packet delay is $l^*-1$.
\item We suppose that SU takes the action 2 when the packet delay is $l<l^*-1$.
\item We have the following inequalities:

{\scriptsize{
\begin{eqnarray*}
\nonumber Q_0(\lambda,l-1)&=&-f(l-1)+V(\Omega^{ns}(\lambda),l),\\
\nonumber &=&-g_u-f(l-1)-c_s+\Omega^{ns}(\lambda)(\phi-P_p+\\\nonumber&&V(\alpha,1))+(1-\Omega^{ns}(\lambda))(\phi-P_{3G}+V(\beta,1)),\\
\nonumber &=&V(\lambda,l)-f(l-1)+(\Omega^{ns}(\lambda)\\\nonumber&&-\lambda)(P_{3G}-P_p+V(\alpha,1)-V(\beta,1)),\\
\nonumber &<&V(\lambda,l).
\end{eqnarray*}}}

The inequality is due to the assumption that $-P_p+V(\alpha,1)<-P_{3G}+V(\beta,1)$ and $\Omega^{ns}(\lambda)\geq\lambda$, and $f(l-1)$ is positive. As the value function is decreasing with the packet delay (see Proposition \ref{monotonicityl}), then $Q_0(\lambda,l-1)<V(\lambda,l-1)+g_u$, i.e. the SU does not take the action $0$ with the packet delay $l-1$.

{\small{
\begin{eqnarray*}
\nonumber Q_1(\lambda,l-1)&=&-c_s+\lambda(\phi-P_p+V(\alpha,1))\\\nonumber&&+(1-\lambda)(-f(l-1)+V(\beta,l)),\\
\nonumber &=&-c_s+\lambda(\phi-P_p+V(\alpha,1))\\\nonumber&&+(1-\lambda)\left(\right.\phi-g_u-f(l-1)-c_s
 +\beta(-P_p\\\nonumber &&\left.+V(\alpha,1))+(1-\beta)(-P_{3G}+V(\beta,1))\right),\\
\nonumber &<&-c_s+\lambda(\phi-P_p+V(\alpha,1))+(1-\lambda)\\\nonumber&&(\phi-g_u-f(l-1)-c_s-P_{3G}+V(\beta,1)),\\
\nonumber &<&-c_s+\lambda(\phi-P_p+V(\alpha,1))\\\nonumber&&+(1-\lambda)(\phi-P_{3G}+V(\beta,1)),\\
\nonumber &=&Q_2(\lambda,l-1).
\end{eqnarray*}}}

The first inequality is due to the assumption that $-P_p+V(\alpha,1)<-P_{3G}+V(\beta,1)$ and the second one is because $g_u$, $f(l-1)$ and $c_s$ are positive. Thus, The optimal strategy is to take action 2 when the packet delay is $l-1$.
Thus, the SU does not take the action $1$ with the packet delay $l-1$. Finally, the SU takes action 2 for all packet delays  and beliefs lower than $\pi(0)$.

\end{itemize}

We now look at the action-value function $Q_2(\alpha,1)$ when the packet delay is $l=1$.
{\small{\begin{eqnarray*}
\nonumber Q_2(\alpha,1)&=&-c_s+\alpha(\phi-P_p+V(\alpha,1))\\\nonumber&&+(1-\alpha)(\phi-P_{3G}+V(\beta,1)),\\
\nonumber Q_2(\alpha,1)&=&\phi-c_s-P_{3G}+V(\beta,1)\\\nonumber&&+\alpha(P_{3G}-P_p+V(\alpha,1)-V(\beta,1)),\\
\nonumber -g_u+Q_2(\alpha,1)&=&-g_u+V(\alpha,1)-P_p+\phi-c_s\\\nonumber&&+(\alpha-1)(P_{3G}-P_p+V(\alpha,1)-V(\beta,1)).
\end{eqnarray*}}}
As the SU takes the action 2 also for the state $(\beta,1)$, we have:
{\small{\begin{eqnarray*}
\nonumber g_u+V(\beta,1)&=&-c_s+\beta(\phi-P_p+V(\alpha,1))\\\nonumber&&+(1-\beta)(\phi-P_{3G}+V(\beta,1)),\\
\nonumber g_u+V(\beta,1)&=&\phi-c_s-P_{3G}+V(\beta,1)\\\nonumber&&+\beta(P_{3G}-P_p+V(\alpha,1)-V(\beta,1)),\\
\nonumber g_u&=&\phi-c_s-P_{3G}+\beta(P_{3G}\\\nonumber&&-P_p+V(\alpha,1)-V(\beta,1)).
\end{eqnarray*}}}
Thus, we obtain:
{\small{
\begin{eqnarray*}
\nonumber -g_u+Q_2(\alpha,1)&=&V(\alpha,1)+P_{3G}-P_p+(\alpha-\beta-1)\\\nonumber&&(P_{3G}-P_p+V(\alpha,1)-V(\beta,1)).
\end{eqnarray*}}}
As we assumed that $P_{3G}-P_p+V(\alpha,1)-V(\beta,1)<0$, and $P_{3G}>P_p$, then we obtain $V(\alpha,1)+g_u\leq Q_2(\alpha,1)$ and therefore the SU takes also the action 2 in the state $(\alpha,1)$. Then we get:
{\scriptsize{
$$
g_u+V(\alpha,1)=Q_2(\alpha,1)=-c_s+\alpha(\phi-P_p+V(\alpha,1))+(1-\alpha)(\phi-P_{3G}+V(\beta,1)).
$$}}
 Let us evaluate finally the difference $V(\alpha,1)-V(\beta,1)$:
{\small{\begin{eqnarray*}
\nonumber V(\alpha,1)-V(\beta,1)&=&(\alpha-\beta)(P_{3G}-P_p+V(\alpha,1)-V(\beta,1)),\\
\nonumber V(\alpha,1)-V(\beta,1)&<&0.
\end{eqnarray*}}}
and
{\small{\begin{eqnarray*}
\nonumber V(\alpha,1)-V(\beta,1)&=&(\alpha-\beta)(P_{3G}-P_p+V(\alpha,1)-V(\beta,1)),\\
\nonumber (V(\alpha,1)-V(\beta,1))(1-\alpha+\beta)&=&(\alpha-\beta)(P_{3G}-P_p),\\
\nonumber V(\alpha,1)-V(\beta,1)&=&\frac{(\alpha-\beta)(P_{3G}-P_p)}{1-\alpha+\beta},\\
&>&0.
\end{eqnarray*}}}
which leads to a contradiction, and therefore, $-P_p+V(\alpha,1)\geq-P_{3G}+V(\beta,1)$. The analysis is similar when $\lambda > \pi(0)$.

\subsection{Proof of Proposition \ref{monotonicityb}}\label{apmonotonicityb}
Let us prove that the value function $V(\lambda,l)$ is increasing with the belief vector $\lambda$ for any packet delay $l$. For all $\lambda_1\leq\lambda_2$, we have that:
{\small{\begin{eqnarray*}
\nonumber V(\lambda_1,l^*)&=&-g_u-c_s+\Phi-P_{3G}+V(\beta,1)\\\nonumber&&+\lambda_1(P_{3G}-P_p+V(\alpha,1)-V(\beta,1)),\\
\nonumber&\leq&-g_u-c_s+\Phi-P_{3G}+V(\beta,1)\\\nonumber&&+\lambda_2(P_{3G}-P_p+V(\alpha,1)-V(\beta,1)),\\
\nonumber &=&V(\lambda_2,l^*).
\end{eqnarray*}}}
This inequality result from the Lemma \ref{cogvs3g}. Let us prove that this propriety holds for all packet delays $l$ using backward induction:
\begin{itemize}
\item Initial condition: There exists a packet delay $l^*$ such that $V(\lambda_1,l^*)\leq V(\lambda_2,l^*),$ $\forall \lambda_1\leq\lambda_2$,
\item We suppose that $V(\lambda_1,l+1)\leq V(\lambda_2,l+1),$ $\forall \lambda_1\leq\lambda_2$,
\item
\paragraph*{First case} We assume that $\Phi+f(l)-P_p+V(\alpha,1)-V(\beta,l+1)\geq0$, then:
\begin{eqnarray*}
\nonumber Q_0(\lambda_1,l)&=&-f(l)+V(\Omega^{ns}(\lambda_1|\theta),l+1),\\
\nonumber&\leq&-f(l)+V(\Omega^{ns}(\lambda_2|\theta),l+1),\\
\nonumber &=&Q_0(\lambda_2,l).
\end{eqnarray*}
The inequality is a direct result from the induction assumption and the Lemma \ref{beliefupdate}. We have also:
\begin{eqnarray*}
\nonumber Q_1(\lambda_1,l)&=&-c_s-f(l)+V(\beta,l+1)\\\nonumber&&+\lambda_1(\Phi+f(l)-P_p+V(\alpha,1)-V(\beta,l+1)),\\
\nonumber&\leq&-c_s-f(l)+V(\beta,l+1)\\\nonumber&&+\lambda_2(\Phi+f(l)-P_p+V(\alpha,1)-V(\beta,l+1)),\\
\nonumber &=&Q_1(\lambda_2,l).
\end{eqnarray*}
\begin{eqnarray*}
\nonumber Q_2(\lambda_1,l)&=&-c_s+\Phi-P_{3G}+V(\beta,1)\\\nonumber&&+\lambda_1(P_{3G}-P_p+V(\alpha,1)-V(\beta,1)),\\
\nonumber&\leq&-c_s+\Phi-P_{3G}+V(\beta,1)\\\nonumber&&+\lambda_2(P_{3G}-P_p+V(\alpha,1)-V(\beta,1)),\\
\nonumber &=&Q_2(\lambda_2,l).
\end{eqnarray*}
The inequalities comes from the Lemma \ref{cogvs3g}. Thus, we have proved that $V(\lambda_1,l)\leq V(\lambda_2,l)$.

\paragraph*{Second case} We suppose that $\Phi+f(l)-P_p+V(\alpha,1)-V(\beta,l+1)<0$, then for all $\lambda$ we have:
\begin{eqnarray*}
\nonumber Q_1(\lambda,l)&=&-c_s+\lambda(\phi-P_p+V(\alpha,1))\\\nonumber&&+(1-\lambda)(-f(l)+V(\beta,l+1)),\\
\nonumber&\leq&-c_s-f(l)+V(\beta,l+1),\\
\nonumber &\leq&-f(l)+V(\beta,l+1),\\
\nonumber &\leq&-c_s-f(l)+V(\Omega^{ns}(\lambda|\theta),l+1),\\
\nonumber &\leq&Q_0(\lambda,l).
\end{eqnarray*}
In fact, we have that $\beta\leq\Omega^{ns}(\lambda|\theta)$ for all belief vector $\lambda$ and the value function $V(\lambda,l)$ is increasing with the belief for the packet delay $l+1$ (induction assumption). Thus, $g_u+V(\lambda,l)=\max{\{Q_0(\lambda,l),Q_2(\lambda,l)\}}$. Moreover, we have:
\begin{eqnarray*}
\nonumber Q_0(\lambda_1,l)&=&-f(l)+V(\Omega^{ns}(\lambda_1|\theta),l+1),\\
\nonumber&\leq&-f(l)+V(\Omega^{ns}(\lambda_2|\theta),l+1),\\
\nonumber &=&Q_0(\lambda_2,l).
\end{eqnarray*}
The inequality is a direct result from the induction assumption. Finally, we have that:
\begin{eqnarray*}
\nonumber Q_2(\lambda_1,l)&=&-c_s+\Phi-P_{3G}+V(\beta,1)\\\nonumber&&+\lambda_1(P_{3G}-P_p+V(\alpha,1)-V(\beta,1)),\\
\nonumber&\leq&-c_s+\Phi-P_{3G}+V(\beta,1)\\\nonumber&&+\lambda_2(P_{3G}-P_p+V(\alpha,1)-V(\beta,1)),\\
\nonumber &=&Q_2(\lambda_2,l).
\end{eqnarray*}
The inequality comes from the Lemma \ref{cogvs3g}.
\end{itemize}
Thus, $V(\lambda_1,l)\leq V(\lambda_2,l)$ for belief vectors $\lambda_1 \leq \lambda_2$ and for all packet delay $l$.
\subsection{Proof of Proposition \ref{thr}}\label{apthr}
In this proposition, we determine explicitly the best action $a^*(\lambda,l)$ for the SU depending on the belief $\lambda$ and the packet delay $l$. At each time slot and for a given information state $(\lambda,l)$, the secondary use will decide to take the action $0$ if $Q_0(\lambda,l)\geq \max{\{Q_{1}(\lambda,l),Q_{2}(\lambda,l)\}}$.
\begin{itemize}
\item First we assume that $Q_1(\lambda,l)>Q_2(\lambda,l)$, then, let us compare $Q_0(\lambda,l)$ and $Q_{1}(\lambda,l)$. The inequality $Q_0(\lambda,l)\geq Q_{1}(\lambda,l)$ is equivalent to:

{\small{\begin{eqnarray*}
\nonumber  -f(l)+V(\Omega^{ns}(\lambda|\theta),l+1) &\geq& -c_s+\lambda(\Phi-P_p+V(\alpha,1))\\\nonumber&&+(1-\lambda)(-f(l)+V(\beta,l+1)), \\
\nonumber V(\Omega^{ns}(\lambda|\theta),l+1) &\geq&  V(\beta,l+1)-c_s+\lambda(f(l)+\\\nonumber&&\Phi-P_p+V(\alpha,1)-V(\beta,l+1)).
\end{eqnarray*}}}

As the value function $V(\lambda,l) $is decreasing with the packet delay $l$ and increasing with the belief $\lambda$, we have $V(\alpha,1)\geq V(\beta,l+1)$. As we assumed that the immediate reward $\phi$ is higher than the  cost $P_p$, we obtain that $f(l)+\Phi-P_p+V(\alpha,1)-V(\beta,l+1)$ is positive. Then, we have the following equivalence:
{\small{\begin{eqnarray*}
Q_0(\lambda,l)&\geq& Q_{1}(\lambda,l) \Leftrightarrow\\ V(\Omega^{ns}(\lambda|\theta),l+1) &\geq&  V(\beta,l+1)-c_s+\lambda(f(l)+\Phi-\\&&P_p+V(\alpha,1)-V(\beta,l+1)).
\end{eqnarray*}}}
Define the functions F and G as follow:
\begin{eqnarray*}
\nonumber F(\lambda,l)&=&{V(\Omega^{ns}(\lambda|\theta),l+1)},\\
\nonumber G(\lambda,l)&=&V(\beta,l+1)-c_s+\lambda(f(l)+\Phi\\\nonumber&&-P_p+V(\alpha,1)-V(\beta,l+1)).
\end{eqnarray*}

We proved in Proposition \ref{pwlc} that the value function is Piecewise linear and convex. Therefore, for all packet delays, the function $F(\lambda,l)$ is PWLC and increasing with $\lambda$ , and the function $G(\lambda,l)$ is linear and increasing with $\lambda$. Note that
\begin{itemize}
\item If $F(\lambda,l)\geq G(\lambda,l)$, then $Q_0(\lambda,l)\geq Q_1(\lambda,l)$ and therefore the best action is $0$.
\item If $F(\lambda,l)< G(\lambda,l)$, then $Q_0(\lambda,l)< Q_1(\lambda,l)$ and therefore the best action is $1$.
\end{itemize}

Let us focus on $F(\pi(0),l)$ and $G(\pi(0),l)$.

Let us prove that $g_u>-f(l)$. We have:
\begin{eqnarray}
\nonumber g_u+V(\alpha,1)\geq Q_0(\alpha,1),\\
\nonumber g_u+V(\alpha,1)\geq -f(l)+V(\Omega^{ns}(\alpha),l+1),\\
\nonumber g_u+V(\alpha,1)-V(\Omega^{ns}(\alpha),l+1)\geq -f(l),\\
\nonumber g_u>-f(l).
\end{eqnarray}
The inequality is because of the monotonicity of the value function and $\Omega^{ns}(\alpha)<\alpha$.
Suppose that the SU chooses the action $0$ for the state $(\pi(0),l)$. We have:
\begin{eqnarray}
\nonumber g_u+V(\pi(0),l)&=&-f(l)+V(\Omega^{ns}(\pi(0)),l+1),\\
\nonumber g_u+V(\pi(0),l)&\leq&-f(l)+V(\Omega^{ns}(\pi(0)),l),\\
\nonumber g_u+V(\pi(0),l)&\leq&-f(l)+V(\pi(0),l),\\
\nonumber g_u&\leq&-f(l).
\end{eqnarray}
This leads to a contradiction as $g_u>-f(l)$.
Thus, $Q_0(\lambda,l)< Q_1(\lambda,l)$ and therefore, $F(\pi(0),l)<G(\pi(0),l)$. Therefore, the cases 1, 3, 5 and 6 are eliminated. Finally, the optimal policy is a kind of threshold and is depicted in the following:
\begin{itemize}
\item The SU  takes the action $0$ for all beliefs lower than the following threshold $$
    Th1(\lambda,l)=\frac{V(\Omega^{ns}(\lambda|\theta),l+1)-V(\beta,l+1)+c_s}{f(l)+\Phi-P_p+V(\alpha,1)-V(\beta,l+1)},$$
     and take the action $1$ otherwise.
\end{itemize}

\item Second, we assume that $Q_2(\lambda,l)>Q_1(\lambda,l)$ and then, we have to compare the action $0$ and $2$, which is equivalent to compare the action-value functions $Q_0(\lambda,l)$ and $Q_{2}(\lambda,l)$. The SU takes the action $0$ instead of the action $2$ if $Q_0(\lambda,l)\geq Q_{2}(\lambda,l)$, which is equivalent to:

{\scriptsize{\begin{eqnarray*}
\nonumber  -f(l)+V(\Omega^{ns}(\lambda|\theta),l+1) &\geq& -c_s+\lambda(\Phi-P_p+V(\alpha,1))\\\nonumber&&+(1-\lambda)(\phi-P_{3G}+V(\beta,1)), \\
\nonumber V(\Omega^{ns}(\lambda|\theta),l+1)&\geq& V(\beta,1)+\Phi+f(l)-c_s-P_{3G}\\\nonumber&&+\lambda(P_{3G}-P_p+V(\alpha,1)-V(\beta,1)).
\end{eqnarray*}}}
We have from the Lemma \ref{cogvs3g}, that $P_{3G}-P_p+V(\alpha,1)-V(\beta,1)\geq 0$. Then, we can provide the same analysis presented in the previous case with the function $F(\lambda,l)={V(\Omega^{ns}(\lambda|\theta),l+1)}$ and the function $G(\lambda,l)=V(\beta,1)+\Phi+f(l)-c_s-P_{3G}+\lambda(P_{3G}-P_p+V(\alpha,1)-V(\beta,1))$. The latter is linear increasing in $\lambda$. We obtain the following threshold policy:
\begin{itemize}
\item The SU takes the action $0$ for all beliefs lower than the following threshold: {\scriptsize{$$
    Th2(\lambda,l)=\frac{V(\Omega^{ns}(\lambda|\theta),l+1)-V(\beta,1)-\Phi-f(l)+c_s+P_{3G}}{P_{3G}-P_p+V(\alpha,1)-V(\beta,1)},$$}}
    and take the action $2$ otherwise.
\end{itemize}

\end{itemize}

\subsection{Proof of Proposition \ref{aftersucc}}\label{apaftersucc}
We have from the Lemma \ref{beliefupdate} that if $\lambda>\pi(0)$ then $\Omega^{ns}(\lambda)\leq\lambda$. Suppose that the SU takes the action $0$ for a belief $\lambda$ and packet delay $l$. Thus we have
\begin{eqnarray*}
\nonumber g_u+V(\lambda,l)&=&-f(l)+V(\Omega^{ns}(\lambda),l+1),\\
\nonumber g_u+V(\lambda,l)&\leq&-f(l)+V(\Omega^{ns}(\lambda),l),\\
\nonumber g_u+V(\lambda,l)&\leq&-f(l)+V(\lambda,l),\\
\nonumber g_u&\leq&-f(l).
\end{eqnarray*}
This leads to a contradiction as $g_u>-f(l)$. The first inequality is because the value function is decreasing with the packet delay and the second one is because that the value function is increasing with the belief and $\Omega^{ns}(\lambda)\leq\lambda$. Thus, if $\lambda>\pi(0)$, then the SU never takes the action $0$ and then $Q_0(\lambda,l)< \max{\{Q_{1}(\lambda,l),Q_{2}(\lambda,l)\}}$.
\subsection{Proof of Proposition \ref{propl}}\label{appropl}
Let us compare the value-action functions  $Q_1(\lambda,l)$ and $Q_2(\lambda,l)$ for all belief vector $\lambda$ and packet delay $l$. The SU waits for next time slot after sensing if $Q_1(\lambda,l)\geq Q_{2}(\lambda,l)$, which is equivalent to:
{\small{\begin{eqnarray*}
   -c_s+\lambda(\Phi-P_p+V(\alpha,1))&&\\+(1-\lambda)(-f(l)+V(\beta,l+1))&\geq&-c_s+\lambda(\Phi-P_p+V(\alpha,1))\\
 &&+(1-\lambda)(\phi-P_{3G}+V(\beta,1)) ,\\
 -f(l)+V(\beta,l+1)\phi-P_{3G}+V(\beta,1)&\geq&0.
\end{eqnarray*}}}
Remark that this condition depends only on the packet delay $l$ and not on the belief vector $\lambda$.


\subsection{Proof of Corollary \ref{corol2}}\label{apcorol2}
If $-f(l)$ is lower than $\Phi-P_{3G}$, then $-f(l)-\Phi+P_{3G}+V(\beta,l+1)-V(\beta,1)$ is always negative. In fact, $V(\beta,2)-V(\beta,1)$ is negative and  $-f(l)-\Phi+P_p+V(\beta,l+1)-V(\beta,1)$ is decreasing with $l$. Therefore, the previous expression is negative for all $l\geq 1$.

\end{document}